\documentclass[10pt,journal]{IEEEtran}
\usepackage{amsmath}
\usepackage{color}
\usepackage{amssymb}
\usepackage{verbatim}
\usepackage{enumitem}
\usepackage{bbm}

\ifCLASSINFOpdf 
\usepackage[pdftex]{graphicx}  %remove draft option to print figures 
\usepackage[pdftex,colorlinks,%draft, %naturalnames,% 
           bookmarks=true,%
                     ]{hyperref}
\usepackage[numbers,sort&compress]{natbib} 
\usepackage{hypernat} 
\else 
\usepackage[dvips]{graphicx} 
\usepackage[dvips, colorlinks,%draft, %naturalnames,% 
           bookmarks=true,%
           ]{hyperref}
 \usepackage[numbers,sort&compress]{natbib} 
\fi

\hypersetup{
   bookmarksnumbered,
   pdfstartview={FitH},
   citecolor={blue},
   linkcolor={red},
   urlcolor=[rgb]{0,0.55,0},%{green},
   pdfpagemode={UseOutlines},
   pdftitle={...}, 
   pdfauthor={A.\ Chattopadhyay,\ B.\ Blaszczyszyn,\ E.\ Altman}
}
 
\IEEEoverridecommandlockouts

\newcommand{\remove}[1]{}

\newcommand{\ind}{\mathbf{1}}
\newcommand{\E}{\mathbf{E}}
\newcommand{\Pro}{\mathbf{P}}

\newtheorem{theorem}{Theorem}

\newtheorem{lemma}{Lemma}

\newtheorem{remark}{Remark}

\usepackage{setspace}
\usepackage{subfigure}

\begin{document}

 \title{Cell planning for mobility management
        in heterogeneous cellular networks
 \thanks{Arpan Chattopadhyay and  Bart{\l}omiej B{\l}aszczyszyn are with Inria/ENS, Paris, France. Eitan Altman is 
 with Inria, Sophia-Antipolis, France. Email: arpan.chattopadhyay@inria.fr, Bartek.Blaszczyszyn@ens.fr,
 eitan.altman@sophia.inria.fr}
% \thanks{The research 
% reported in this paper was supported by...}
% \thanks{{\bf All appendices are provided in the supplementary material.}}
}
% \newcounter{one}
% \setcounter{one}{1}
% \newcounter{two}
% \setcounter{two}{2}
%\vspace{-3cm}
%\author{
%{\singlespacing
%Arpan~Chattopadhyay$^\fnsymbol{one}$, and Anurag~Kumar$^\fnsymbol{one}$\\
%\vspace{0.3cm}
%{\singlespacing
%\parbox{0.49\textwidth}{\centering $^\fnsymbol{one}$Dept. of ECE, Indian Institute of Science\\
%Bangalore 560012, India\\
%arpanc.ju@gmail.com, anurag@ece.iisc.ernet.in}
%\hfill
%}

\author{
%\vspace{-0.5cm}
%{\singlespacing
Arpan~Chattopadhyay, Bart{\l}omiej~B{\l}aszczyszyn, Eitan~Altman \\
%\vspace{0.3cm}
%{\singlespacing
% \parbox{0.49\textwidth}{\centering $^\fnsymbol{one}$Dept. of ECE, Indian Institute of Science\\
% Bangalore 560012, India\\
% arpanc.ju@gmail.com, anurag@ece.iisc.ernet.in}
% \hfill
% \parbox{0.49\textwidth}{\centering $^\fnsymbol{two}$Telecom ParisTech and CNRS LTCI \\
% Dept. Informatique et R\'eseaux\\
% 23, avenue d'Italie, 75013 Paris, France\\
% marceau.coupechoux@telecom-paristech.fr}
%}
}

\maketitle
\thispagestyle{empty}

\begin{abstract}
%In this paper, we address the problem of time-average downlink data
%rate maximization of mobile users in cellular networks. 
In small cell networks, high mobility of users 
results in frequent handoff and thus severely restricts the data rate for mobile users. 
To alleviate this problem, we
propose to use heterogeneous, two-tier  network structure 
where static users are served by both macro and micro base stations,
whereas the mobile (i.e., moving) users are served only by macro base stations having
larger cells; the idea is to prevent frequent data outage for mobile
users due to handoff.
We use the classical two-tier Poisson network model with different
transmit powers (cf~\cite{DHILLON2012}), assume independent Poisson process of
static users and doubly stochastic Poisson process of mobile users 
moving at a constant speed along infinite straight lines generated by
a Poisson line process.
Using stochastic geometry, we calculate the  average downlink {\em data
  rate}  of the typical static and mobile (i.e., moving) 
users, the latter accounted for handoff outage periods. We consider
also the average {\em throughput} of these two types of users defined as
their average data rates divided by the mean total number of users
co-served by the same base station. We find  that  if 
 the density of a homogeneous network and/or the speed of mobile users is high, 
it is advantageous to let the mobile users
connect only to some optimal fraction  of BSs to reduce the frequency
of handoffs during which  the connection is
not assured. If a heterogeneous structure of the  network is allowed,
 one can further jointly optimize the mean throughput of mobile and static  users by 
appropriately tuning the powers  of micro and macro base stations
subject to some aggregate power constraint ensuring unchanged  mean data rates of
static users via the network equivalence property (see \cite{netequivalence}).

\end{abstract}

\begin{keywords}
Heterogeneous network, small cells, Poisson routes, mobility,  handoff, cell planning,
spectral efficiency, throughput, resource allocation, stochastic geometry.
\end{keywords}

\section{Introduction}\label{section:introduction}
The proliferation of high-specification handheld/mobile devices such as smartphones and tablets has led to unprecedented 
growth in cellular traffic over the past few years, and is expected to grow continuously even further. These devices are 
often equipped with 3G or 4G communication capabilities, and are able to run applications such as video streaming or downloading, image or 
media file transfer via e-mail, social networking applications, and access to several cloud services. 
In order to meet the enormous bandwidth demand for these applications, the use of small cell networks 
(see \cite{andrews-etal12femtocells}) assisted by macro base stations (see \cite{pauli-etalYYhet-LTE-ICI-coordination}, 
\cite{stanze-weber13het-LTE-A}, \cite{nakamura-etal13trends-small-cell-LTE}, \cite{ishii-etal12novel-architecture-LTE-B}) have 
recently become popular; the major idea behind such heterogeneous network architecture is that the small cells 
(such as femtocells and picocells) can meet the bandwidth demand of the users, while the macro base stations provide cellular coverage.

Small cell networks (e.g., picocell networks) can provide high throughput to the static users,
but may significantly deteriorate the performance of mobile (i.e., moving) users.
Indeed, very high mobility of users in small cell networks 
result in frequent handoff at  cell boundaries, thereby resulting in 
potentially huge signaling overhead 
(see \cite{camp-etal02survey-mobility-models}, \cite{halepovic-williamson05mobility-cellular-data-network}, 
\cite{akyildiz-etal99mobility-management-next-generation-wireless})
among the base stations and mobile users. 
In case of hard handoff, the existing connection to the base station is terminated before the connection 
to a new base station is established. As a result of this and the signaling overhead, during 
some short time period for each hard handoff, 
the mobile user is not able to receive any data from the base station 
(\cite[Section~$3$]{kavitha-etal11spatial-queueing-picocells} says that a moving user requires some fixed 
amount of communication with the base stations 
for each handoff; this can be modeled as a temporary outage for the desired data transmission). This happens 
even if we assume that every hard handoff attempt is successful, i.e., the target cell can always accommodate a new 
handoff request.\footnote{LTE does not use any soft handoff.} 
In case of soft handoff, the mobile user maintains connection to more than one base stations during the 
handoff period (which is a significant amount of time), thereby reducing the spectral efficiency (apart 
from the throughput loss due to signaling overhead). 
As a consequence, if the mobility of users is very high and/or network
cells are too small, then a significant loss in average data-rate can
be experienced by the mobile users. 
A solution to the data outage problem (due to hard handoff)  
is predictive channel reservation as described in \cite{ye-etal06predictive-channel-reservation-handoff-prioritization}, where handoff 
requests are sent to adjacent cells when the mobile users reach within
a given distance of the cell boundary depending on rather complicated 
user location and velocity estimation~\footnote{GPS tracking of users is not allowed by the legislation  of some countries.}. 
This idea is useful in preventing 
the temporary data outage due to handoff, but it is unable to address the
problem of frequent handoff and consequently huge signaling overhead rate due to
high mobility of users.

As a solution to the above problem, we propose to use a heterogeneous network architecture, where only macro base stations can serve the 
mobile users; the relatively large cell size of the macro base stations result in a much smaller handoff rate in this architecture. However, 
macro base stations are less power-efficient as compared to micro base stations, and micro base stations 
provide high throughput to the static users. Clearly, there is a trade-off between the fraction of macro and micro base stations 
to be chosen by the network service provider, and also between the transmit power levels at which these two classes of base 
stations should operate. In this paper, under certain modeling assumptions, we provide a stochastic geometry framework which can be used to find the optimal 
fraction of macro base stations and the transmit power levels of macro and micro base stations, as a function of the densities of static and mobile users and the 
velocity of mobile users.
Our main findings  in this matter are as follows.
\begin{itemize}
\item In a homogeneous network scenario, if the product of the mobile
  user speed, handoff time and the square root of the network density  
is large then  it is  advantageous to let the mobile users
connect only to some fraction  of base stations to reduce the frequency
of handoffs during which their data-rate drops down. Obviously too
small such fraction  will result is low data rate 
due to large distance from base stations. Our model
allows us to calculate the optimal value of this fraction 
in function of mobile speed and other network parameters.
\item If a heterogeneous structure of the  network is allowed, it is
  possible to further  optimize  the  mean  throughput  of mobile and  static  users; 
  a good compromise between the performance of static and mobile users is obtained by appropriately choosing the transmit power levels and 
  densities of macro and micro base stations.
\end{itemize}

\subsection{Related Work}\label{subsection:related_work}
There has been a lot of work in the literature on the impact of mobility of users in wireless networks. In fact, the authors in 
\cite{grossglauser-tse02mobility-increases-capacity} have shown that mobility increases the capacity of ad-hoc networks. 
\cite{bansal-liu03capacity-delay-mobility-ahn} deals with the trade-off of delay and throughput in ad-hoc networks in presence of mobility. 
The authors of \cite{bonald-etal09flow-level-performance-mobility}, 
\cite{bonald-etal04mobility-flow-level-data-systems}, \cite{borst-etal12capacity-with-mobility},  
\cite{orlik-rappaport01handoff-arrival-process-cellular} discuss the impact of inter and intra cell mobility on capacity, flow level performance and 
the trade-off between throughput and fairness; these results show that mobility increases the capacity of cellular networks when base stations 
interact among themselves, since the cooperation among base stations allows the users to improve performance from multi-user diversity, 
opportunistic scheduling etc. 

However, in practice, handoff 
results in an outage in connection and high signaling overhead, which none of the mentioned references care about. 
As a result, researchers have recently focused on analyzing the impact of user mobility on call block and call drop probabilities and 
optimal cell size in picocell networks deployed {\em on a line}; 
see  \cite{ramanath-etal10spatial-queueing-mobility-wiopt}, \cite{kavitha-etal11spatial-queueing-picocells},  
\cite{ramanath-etal10mobility-small-cell}. 

Unlike these papers, we consider a macro-assisted small cell network 
{\em on the two-dimensional plane} where only macro base stations 
are allowed to serve the mobile users; our goal is to choose the network design parameters in such a way that the time-average of a linear combination 
of the rates of static and mobile users is maximized, depending on the handoff duration and user velocity. 
Handoff control by macro-assisted small cell networks has been proposed before (see \cite{lee-syu14handover-macro-assisted-small-cell-networks}, 
\cite{ishii-etal12novel-architecture-LTE-B}), but no in-depth mathematical analysis was provided that can be used as a guideline to choose 
network design parameters such as the density of macro base stations and transmit power levels of the two classes of base stations.

\subsection{Our Contribution}\label{subsection:our_contribution}
Our contributions in this paper are summarized as follows:
\begin{itemize}
\item In stochastic geometric analysis of Poisson multi-tier network, 
  Theorem~\ref{lemma:ccdf-sinr-tau-less-than-1} provides a simplified expression
  for the SIR coverage probability by one tier for the whole domain of
  SIR.  It is based on the explicit evaluation of the integral
  function $\mathcal{I}_{n,\beta}(\cdot)$
  defined in~\cite[Equation~$(13)$]{bartek-keeler15sinr-process-poisson-networks-factorial-moment-measures}
  and capturing  the impact of the noise, when this latter  has the distribution of a  Poisson-shot noise
 variable (here the interference from the non-serving
 tier).~\footnote{The numerical evaluation of the original integral
   $\mathcal{I}_{n,\beta}(\cdot)$ is quite tricky, as reported
   in~\cite{keeler-programs}.}
This expression is used to evaluate and optimize the (static and
mobile) user mean downlink bit-rates.
 \item We provide explicit or integral expressions for several
   primitive characteristics of the two-tier Poisson network with
   Poisson-line road system, including mean areas of different types
   of zero-cells (cells covering the origin; they are statistically
   larger than the respective typical cells), mean number of static
   and mobile users served in these cells, the intensity of cell boundary
   crossings (handoffs) for a user traveling along a  given line.
We believe these expressions 
 are new to the wireless literature.
\item Using the above expressions we evaluate and optimize (over the densities and transmit power levels of the macro and micro base stations) the
  mean throughput of static and mobile users, the latter accounted
  for handoff outage periods. The mean throughput of a typical (static
  or mobile) user is defined as the mean bit-rate of this user divided
  by the mean number of all users co-served by the same station.

\end{itemize}

\subsection{Organization}\label{subsection:organization}
The rest of the paper is organized as follows. A two-tier Poisson
network model with Poisson-line route system  for
stochastic geometric analysis of the performance of static and mobile
users is 
provided in Section~\ref{section:system_model_and_notation}. 
In Section~\ref{ss.metrics}, using this model,
we evaluate the mean data rate and
throughput  of the typical mobile and static user. 
Next, in
Section~\ref{section:optimal-design-of-the-network-using-stochastic-geometry}
we formulate and solve the problem of the optimal  network design
with respect to these performance metrics of static and mobile users. 
Finally, we conclude in Section~\ref{section:conclusion}. All proofs are provided in the appendix.

\section{Heterogeneous Poisson Network Model with Poisson Road System}\label{section:system_model_and_notation}

\subsection{System Assumptions}
In this paper, we consider  radio part of the downlink traffic
in a network consisting of two types of base stations (BSs) using constant
but possibly different powers. The coverage in the network is ensured
by {\em macro stations} which typically transmit with larger powers. 
{\em Micro stations}, which typically transmit with smaller powers, 
are used to meet the capacity request. 
Both types of BSs transmit over the entire available bandwidth.
The above assumptions are satisfied e.g. for practical cellular
network, such as~LTE~networks.

We consider also two types of users: static and mobile ones, the latter moving along some road system with constant
velocity. We think, for example, of a urban or suburban road system with  users sitting in
moderately fast moving cars and downloading contents from the base stations.
In general, all users are served by the respective  BSs received with the
strongest powers. However  mobiles users, which are  subject to handoff procedures when
changing the serving stations, are allowed not to connect to micro
BSs, the goal being to reduce the handoff frequency. 
During the handoff event, which takes some fixed time, the quality of
the mobile users connection drops down (in particular due to more
heavy signaling) or even suffers  a temporary
outage period. We assume that the applications being run in the user
ends are enough delay
tolerant and that they  resume perfectly once the communication with a new base
station is established immediately after handoff is over.

Regarding the radio channel assumptions, we consider interference
limited scenario, where the  peak data rate available at any location is limited by the
interference from the  base stations not serving this location. Again
this is a reasonable scenario for urban and suburban environment.

\subsection{Two Tier Network Model}\label{subsection:base-station}
We consider a heterogeneous network (hetnet) model composed of two tiers of
base stations (BSs) modeled by two independent Poisson point processes,
$\Phi_{macro}$ and $\Phi_{micro}$ of intensity $p\lambda_{BS}$
and $(1-p)\lambda_{BS}$, respectively, where $p\in[0,1]$ is the fraction of macro BSs.
Macro and micro BSs transmit with constant but possibly different powers 
$P_{macro}, P_{micro}$ respectively. It is natural to assume  $P_{macro}>P_{micro}$. 
These powers, together with $p$,  are network design parameters
subject to some network equivalence condition that will be explained in Section~\ref{ss.net-equivalence}. 
Note that the superposition  of two tiers
$\Phi:=\Phi_{macro}+\Phi_{micro}$ forms a Poisson point process of
intensity $\lambda_{BS}$ and $p$ is the fraction of BSs which are
macro stations. We will denote the  locations of BS in the two tiers by  $\Phi=\{X_i\}$, $X_i\in\mathbb{R}^2$,
with arbitrary, countable  labeling of BSs by index $i$.

\subsection{Static Users  and Mobile Users on the Routes}
\label{subsection:mobile-users}
We consider two classes of users. Locations of {\em static users (SU)} are modeled by points of a Poisson point process $\mathcal{U}_{static}$
of intensity $\lambda_{SU}$. {\em Mobiles users (MU)} are moving with
same constant speed $v$ on a system of directed routes  modeled by
straight lines of a directed homogeneous  Poisson line process
$\mathcal{L}$ on the plane, of intensity $\lambda_L$ that  corresponds
to the mean total length of routes per unit of surface, (see 
\cite[Chapter~$8$]{chiu-etal13stochastic-geometry-and-its-applications}).
%This is a reasonable model for roads in a city where users sitting in fast moving cars download contents from the base stations.
We assume that, given a realization of the Poisson line process, at time instant $t=0$, MUs form a Poisson point process
$\mathcal{U}_{mobile}$ on $\mathcal{L}$ of intensity $\lambda_{MU}$ MUs per unit of 
route length.~\footnote{More formally, the process $\mathcal{U}_{mobile}$ is a doubly 
stochastic Poisson point process with random intensity $\lambda_{MU}\mathcal{L}$.}
This means, in particular,  that any two  successive MUs on any line
of $\mathcal{L}$ are separated by a distance having exponential
distribution with mean $1/\lambda_{MU}$. Moreover, constant mobility
of MUs implies that at any time instance $t$ the relative locations
(and hence  distribution) of MUs on the lines of $\mathcal{L}$ remain
unchanged. Also, one observes MUs crossing any point of a line 
according to a time homogeneous Poisson process with rate $\lambda_{MU} v$.
We assume that all processes $\Phi_{macro}$, $\Phi_{micro}$, $\mathcal{U}_{static}$ and  $\mathcal{U}_{mobile}$ are independent. 
A sketch of the network model with the both types of users is depicted in Figure~\ref{fig.hetnet-with-lines}.

\begin{figure}[t!]
\begin{center}
\begin{minipage}{1\linewidth}
\begin{center}
\centerline{\includegraphics[width=0.9\linewidth]{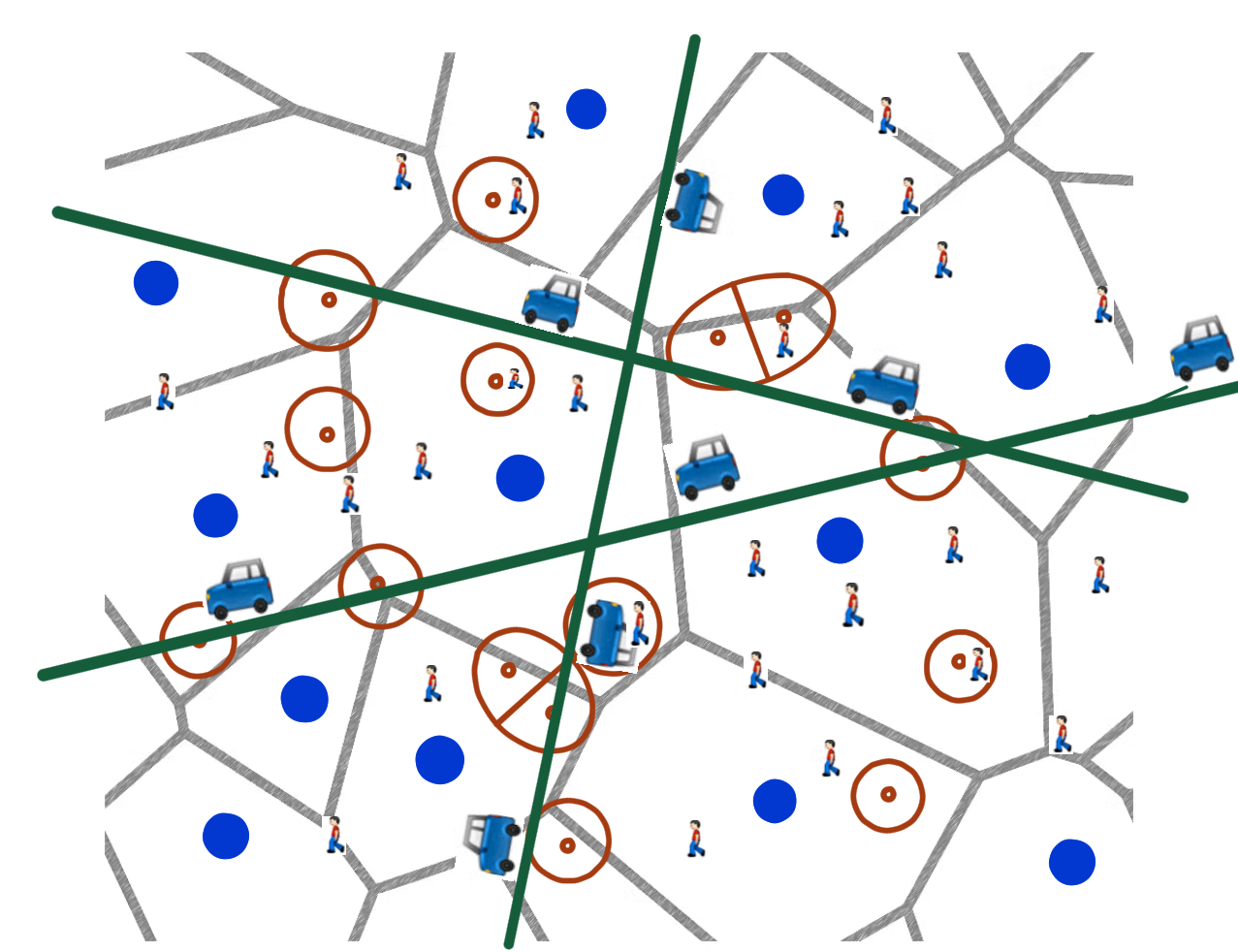}}
\vspace{-2ex}
\caption[Two-tier heterogeneous network with straight line route system]{Two-tier heterogeneous network with straight line route system.
Statics users are distributed everywhere, whereas mobile users are moving along the routes. The blue circles correspond to macro BSs, and the tiny 
dots correspond to micro BSs. The regions around the micro BSs correspond to the regions where a micro BS sends the strongest signal. 
The Voronoi tessellation is generated solely by the process $\Phi_{macro}$. Static users are shown using human symbols, and the mobile 
users (moving along straight lines) are shown using car symbols.
\label{fig.hetnet-with-lines}}
\end{center}
\end{minipage}
\end{center}
\vspace{-2ex}
\end{figure}

\subsection{Downlink Wireless Channel Model}\label{subsection:channel-model}
The path-loss at a distance $r$ from a  BS is given by $(Ar)^{\beta}$, where $A>0$ and $\beta>2$ are two constants. 
We ignore all  random propagation effects (fading, shadowing) in the channel. Moreover, 
we focus on the interference limited regime ignoring any thermal noise at the user receivers.

\subsection{Downlink Service Discipline}\label{subsection:service-discipline}
Static users are served by the respective  BSs in $\Phi$ which are received with the strongest power
~\footnote{With any tie-breaking rule; in fact the probability that a user receives the same power from two or more  different stations of $\Phi$ is null.}. 
This means that a static user located at $x$ on the plane  
is served by the station $X_i\in\Phi$ which maximizes
 the value 
%$$\frac{P_{macro}\ind(X_i\in\Phi_{macro})+P_{micro}\ind(X_i\in\Phi_{micro})}{(A|X_i-x|)^{\beta}}\,.$$
$P_{i}(A|X_i-x|)^{-\beta}$, where $P_i=P_{macro}$ if
 $X_i\in\Phi_{macro}$ and $P_i=P_{micro}$ if
 $X_i\in\Phi_{micro}$. 
Mobile users are served only by macro  BSs in $\Phi_{macro}$, with the
choice of the serving station also based on the  stronger received
power. Due to constant power emitted by macro BSs, this is equivalent
to saying that mobile users are served by the respective closest macro BS.  

For any micro or macro BS, by its {\em hetnet cell} we call  the subset of the plane where this station serves static users.
For macro BSs, besides their  hetnet cells, we consider also {\em macro cells}; these are Voronoi cells generated by $\Phi_{macro}$ 
(ignoring $\Phi_{micro}$), where the macro BSs serve mobile users. Note that the hetnet cells generated by  the macro BSs are subsets of the respective macro (Voronoi) cells of these BSs. Also, if 
$P_{macro}>P_{micro}$ then the hetnet cells of macro BSs are statistically larger than the hetnet cells of the micro BSs.
This situation is depicted via Figure~\ref{fig.hetnet-with-lines}.

\subsection{Handoff in the Network}\label{subsection:handoff-model}
As a MU moves along a line $l \in \mathcal{L}$, it traverses across various macro cells 
(Voronoi cells generated by $\Phi_{macro}$).
On the macro cell boundaries, one macro BS has to hand over the MU to the BS of the neighboring cell;
this event is called handoff. We assume in this paper that handoff is always successful (we will discuss in the conclusion 
how the possibility of handoff failure due to overload in the target cell 
can be taken care of in our current framework). However, 
a significant signaling involved 
during the handoff seriously affects the (downlink) transmission rate during the handoff (see 
\cite[Section~$3$]{kavitha-etal11spatial-queueing-picocells}; a MU requires some fixed amount of signaling/overhead 
communication with the base stations 
for each handoff). In case of hard handoff, the existing connection to the base station is terminated before the connection 
to a new base station is established. As a result of this and the signaling overhead, during 
some short time period for each hard handoff, 
the mobile user is not able to receive any data from the base station. In case of soft handoff, 
the mobile user maintains connection to more than one base stations during the 
handoff period, thereby reducing the spectral efficiency (apart 
from the throughput loss due to signaling overhead). 
In order to account for the throughput loss due to handoff in a simplified way,  we assume that
during handoff, 
the MU is not able to receive any data from either of the two neighbouring  base stations for a constant time $T_h$.

Note that MUs do not stop, but keep moving with the usual speed $v$  during 
the handoff event. Again for simplicity, we assume that the segment of the line $l$  (of length $v T_h$) traversed during the handoff period 
is bisected by the traversed cell boundary.

Note that, if the density of macro BSs is high  (and therefore the macro cells sizes are small) with respect to the 
MU speed $v$, then frequent handoff events have serious detrimental effect on the average downlink  data rate of MUs. 
The reason that  we allow only the macro base stations to serve the MUs is to reduce handoff rate by using only large cells for MUs. 
The goal of the hetnet optimization in $p,P_{macro},P_{micro}$ considered 
in Section~\ref{section:optimal-design-of-the-network-using-stochastic-geometry} is to optimize the performance of the MUs 
while (at least) preserving the performance of the static users.

\subsection{Network Equivalence}
\label{ss.net-equivalence}
When optimizing the network design in $p,P_{macro},P_{micro}$, we will consider the following 
constraint 
\begin{equation}\label{e.hetnet-equivalence}
pP_{macro}^{2/\beta}+(1-p)P_{micro}^{2/ \beta}=P^{2/ \beta}\,,
\end{equation}
where $P$ is some given fixed transmission power.
Condition~\eqref{e.hetnet-equivalence} ensures that the interference field over $\mathbb{R}^2$ generated by the hetnet
will have the same marginal distributions as the homogeneous Poisson
network of density $\lambda_{BS}$, where each base station transmits
at fixed power $P$.  For details, see the notion of {\em equivalent homogeneous network} as 
explained in \cite{netequivalence,bartek-keeler15sinr-process-poisson-networks-factorial-moment-measures}. 
The equality of the marginal distributions means that all static users experience the same mean service characteristics
based on the collection of powers they receive from all macro and micro stations as 
in the equivalent homogeneous network.
\footnote{ Note that the constraint~\eqref{e.hetnet-equivalence} is  different from the 
constraint on the mean transmit power per base station $pP_{macro}+(1-p)P_{micro}=P$.}

\section{Performance Evaluation of the Heterogeneous Network}
\label{ss.metrics}
In what follows we shall evaluate the performance of mobile and static
users in our hetnet model. We consider mean  {\em bit rates} of
a single, typical SU or MU, the latter accounted for handoff outage events.
For both types of users, we also consider the {\em mean throughput}, which is defined as
the mean peak bit rate divided by the mean number of all (static and mobile)
users served by the station serving the typical user.
The above performance metrics are subject to network design
optimization in Section~\ref{section:optimal-design-of-the-network-using-stochastic-geometry}.

\subsection{Downlink Bit Rate of Mobile Users} 
Denote by $X^*$ the macro BS that  is closest to the origin; it is the 
BS serving the typical MU present at the origin of the plane.
Denote by  $\E[R_{macro}(0)]$ the {\em mean downlink (Shannon) bit-rate}\footnote{In fact it is the spectral efficiency, i.e., the bit-rate per unit of bandwidth.} at
  the origin from $X^*$, where  
$R_{macro}(0)=\log_2(1+\text{SIR}_{macro}(0))\,,$ 
with 
$$\text{SIR}_{macro}(0):=\frac{P_{macro}(A|X^*|)^{-\beta}}{\sum_{X_i\in\Phi, X_i\not=X^*}P_i(A|X_i|)^{-\beta}}\,,$$
where $P_i=P_{macro}$ if
 $X_i\in\Phi_{macro}$ and $P_i=P_{micro}$ if
 $X_i\in\Phi_{micro}$.
Observe that MUs treat the power received  form micro BSs  as noise.
We have 
\begin{equation}\label{e.ER}
\E[R_{macro}(0)]=
%\frac{1}{2}\int_{0}^{\infty}\Pro \bigg(\log_2(1+\text{SIR})>t \bigg)dt=
%\frac{1}{2} \int_{0}^{\infty}\Pro \bigg(\text{SIR}>(2^{t}-1) \bigg)dt= \frac{1}{2} \int_{0}^1\Pro \bigg(\text{SIR}>(2^{t}-1) \bigg)dt+
\int_{0}^{\infty}\Pro\{\text{SIR}_{macro}(0)>(2^{t}-1)\}\,dt
\end{equation}
and the distribution of $\text{SIR}_{macro}(0)$ 
is the subject of the following result.
%can be evaluated along the
%line presented  in \cite[Equation~$(16)$]{bartek-keeler15sinr-process-poisson-networks-factorial-moment-measures}
\begin{theorem}\label{lemma:ccdf-sinr-tau-less-than-1}
 For any $\tau>0$, we have:
  \footnotesize
 \begin{eqnarray*}
  && \Pro(\text{SIR}_{macro}(0)>\tau)\\
  &=& \sum_{n=1}^{\lceil{\frac{1}{\tau}}\rceil}(-1)^{n-1}\bigg( \frac{\tau}{1-(n-1)\tau} \bigg)^{-\frac{2n}{\beta}} \mathcal{J}_{n, \beta} \bigg( \frac{\tau}{1-(n-1)\tau} \bigg) \\
  && \times \frac{\beta}{2} \times \bigg( \frac{2}{\beta \Gamma(1-\frac{2}{\beta}) \Gamma(1+\frac{2}{\beta})} \bigg)^n  \bigg( 1+\frac{(1-p)P_{micro}^{\frac{2}{\beta}}}{p P_{macro}^{\frac{2}{\beta}}} \bigg)^{-n}
 \end{eqnarray*}
 \normalsize
where $\mathcal{J}_{n,\beta}(x)$ is defined in Appendix~\ref
{appendix}, Subsection~\ref{subsection:expression-for-Jn} 
(taken from \cite[Equation~$(16)$]{bartek-keeler15sinr-process-poisson-networks-factorial-moment-measures} 
with $x_1=x_2=\cdots=x_n=x$).
In particular, for $\tau\ge 1$, we have:

\footnotesize
$$\Pro (\text{SIR}_{macro}(0)>\tau)=\bigg[ \tau^{\frac{2}{\beta}} \Gamma(1-\frac{2}{\beta}) \Gamma(1+\frac{2}{\beta}) 
 \bigg( 1+\frac{(1-p)P_{micro}^{\frac{2}{\beta}}}{p
   P_{macro}^{\frac{2}{\beta}}} \bigg)  \bigg]^{-1}.$$
\normalsize
\end{theorem}
\begin{proof}
 See Appendix~\ref
{appendix}, Subsection~\ref{subsection:proof-of-lemma-ccdf-sinr-tau-less-than-1}.
\end{proof}

\begin{remark}\label{rem.SIR-equivalence}
\begin{enumerate}
\item Under the equivalent network condition~\eqref{e.hetnet-equivalence}, we have: 
$$1+\frac{(1-p)P_{micro}^{\frac{2}{\beta}}}{p P_{macro}^{\frac{2}{\beta}}}=\frac{P^{\frac{2}{\beta}}}{p P_{macro}^{\frac{2}{\beta}}}.$$ 
\item In case of homogeneous network $P_{macro}=P_{micro}=P$ with
  only a fraction $p$ of stations potentially  serving the origin,
  the above quantity becomes equal to~$1/p$. Then, for $\tau \geq 1$, the 
probability $\Pro (\text{SIR}>\tau)$ increases linearly in $p$. 
\item Further assuming $p=1$ one obtains the coverage probability in
  the so called {\em equivalent homogeneous network}. It is equal to the coverage probability of
  the typical static user connecting to the strongest station (macro
  or micro) in the  hetnet, cf. Section~\ref{ss.SU-throughput}.
\end{enumerate}
\end{remark}

\begin{remark}\label{rem:ergodicity-along-a-line-without-handoff-for-a-single-user}
By the ergodicity of the model, the expectation $\E[R_{macro}(0)]$ is equal almost surely to the
sample average data rate along  any line of the Poisson line process $\mathcal{L}$.
 See Appendix~\ref{appendix},
 Subsection~\ref{subsection:proof-of-ergodicity-theorem} for more  explanations.
\end{remark}

\subsection{Accounting for handoff}
\label{ss.rate-handoff}
Note that $\E[R_{macro}(0)]$ does not account for the handoff   outage
  events. An exact way of taking into account this latter phenomenon
  would require calculating $\E^0_{MU}[R_{macro}(0)\ind(0\ \text{not in handoff})]$ 
  (where $\E^0_{MU}$ is the expectation w.r.t. the Palm probability that a mobile user is located at the origin), 
  which is not amenable to explicit analysis, in particular because of
the  dependence between $R_{macro}(0)$ and the event $\{0\ \text{not in
    handoff}\}$. Regarding the  handoff probability we have the
  following bound that involves the  intensity  
$$\lambda_c=\frac{4 \sqrt{\lambda_{BS}p}}{\pi}$$
of crossings of a fixed
 straight line with the boundaries of the macro 
 cells, which are Voronoi cells of $\Phi_{macro}$,
 cf~\cite[Equations~$5.7.4$ with $m=2$]{okabe99spatial-tesselations}
~\footnote{Since $\Phi_{macro}$ is motion invariant, $\lambda_c$ is invariant with respect to the choice of the fixed
  line.}.
\begin{lemma}\label{lemma:lower-bound-on-fraction-of-nonoutage-length-over-real-line}
$\Pro_{MU}^0\{0\ \text{not in handoff}\}\ge (1-\lambda_c vT_h)$.
\end{lemma}
\begin{proof}
See
Appendix~\ref{appendix}, 
Subsection~\ref{subsection:proof-of-lower-bound-on-fraction-of-nonoutage-length-over-real-line}.
\end{proof}
The above bound is meaningful only for $\lambda_cvT_h$ smaller than~1
and  tight when it is close to 0. With the above precautions, for the
sake of analytical tractability, we 
will consider the product 
$(1-\lambda_c vT_h)\E[R_{macro}(0)]$  
as a substitute for the typical MU bit-rate accounted for handoff outage.  

\subsection{Accounting for Other Users --- Mean Throughput of MUs} 
\begin{figure}[t!]
\begin{center}
\begin{minipage}{1\linewidth}
\begin{center}
\centerline{\includegraphics[width=0.4\linewidth]{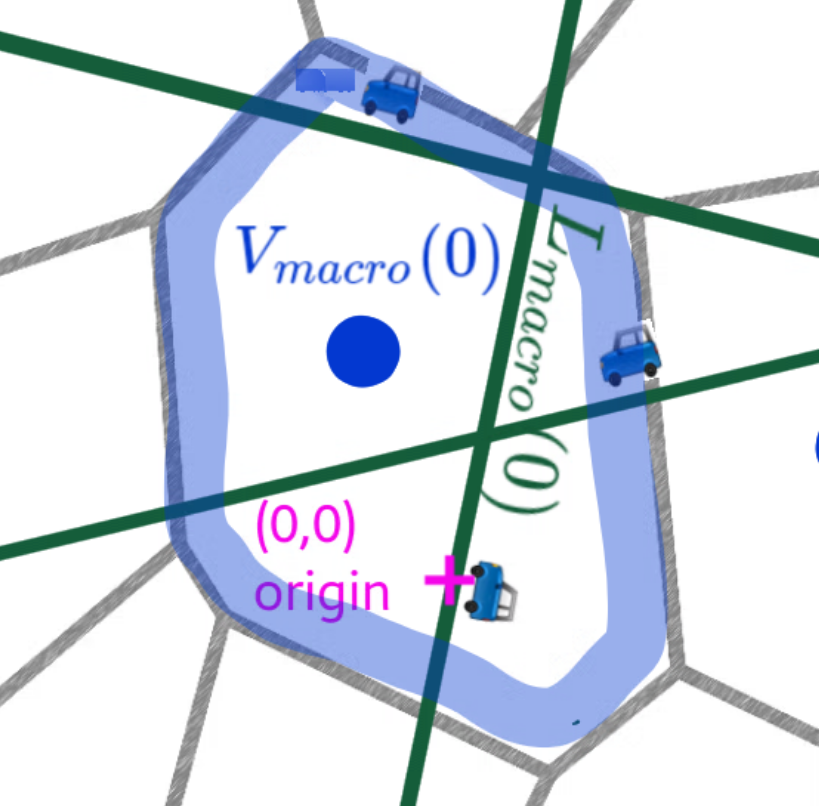}
\hspace{0.03\linewidth}
\includegraphics[width=0.4\linewidth]{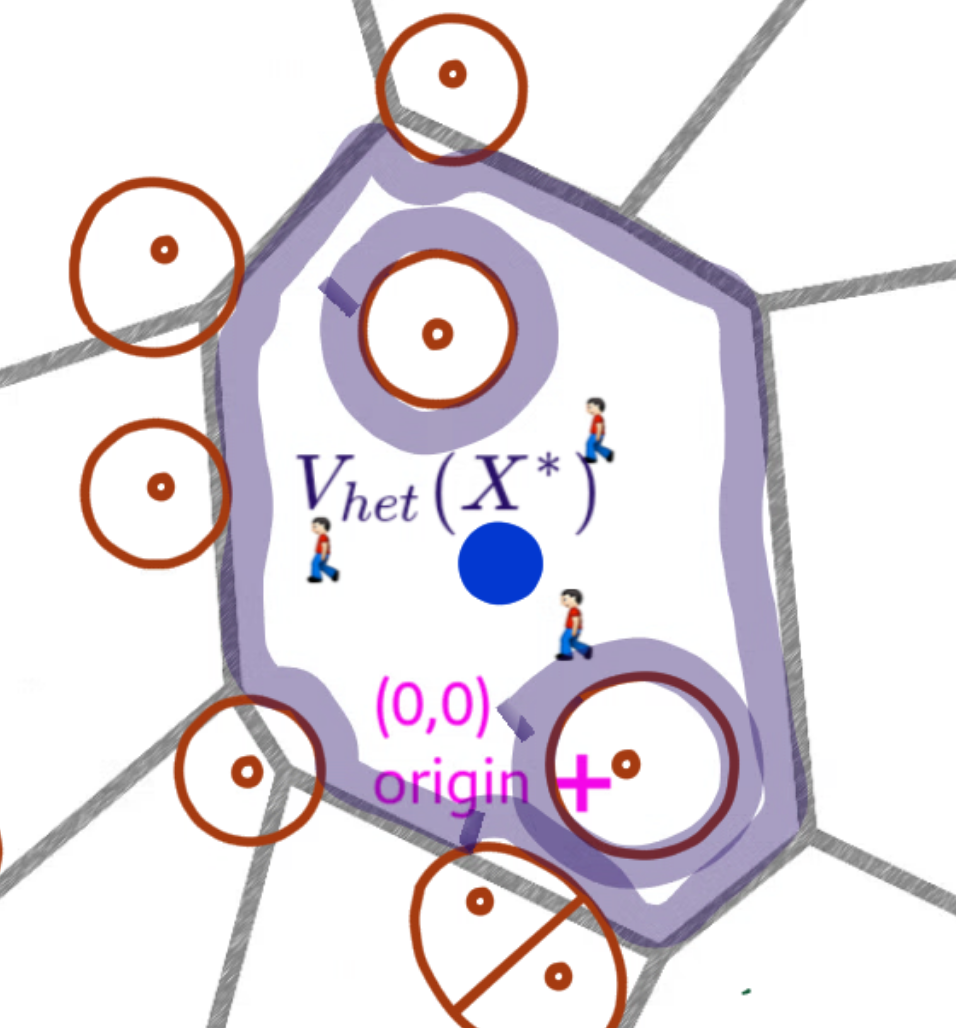}}
\vspace{-2ex}
\caption[Zero macro cell and its zero hetnet subset]{Left: The  macro BS $X^*$ closest to the origin and its macro-cell,
called zero-macro cell. Right: the hetnet cell of $X^*$. Note it is not necessarily the hetnet cell covering the origin.
\label{fig:hetnet-zero-macro-cel}}
\end{center}
\end{minipage}
\end{center}
\end{figure}

The expression $(1-\lambda_c vT_h)\E[R_{macro}(0)]$ is only the mean  bit-rate of a
single MU served by $X^*$  and  does not account
for the fact that $X^*$ needs to share its resources with other MUs
and SUs. 

In order to account for the rate sharing with other users, denote 
by $V_{macro}(0)$  the macro cell of $X^*$, and by $|V_{macro}(0)|$ the area of this macro cell. This is the
macro cell (generated by $\Phi_{macro}$ alone)  covering the origin; cf Figure~\ref{fig:hetnet-zero-macro-cel}~left. Following standard stochastic
geometric terminology we call it  {\em zero macro  cell}.
By $\overline N_{MU,macro}:=\E^0_{MU}[| \mathcal{U}_{mobile}(V_{macro}(0)) |]$
we denote the mean number of MUs present in this zero macro cell 
under the Palm distribution for MUs (i.e., given the typical MU at
the origin).

Moreover, let $V_{het}(X^*)$  be the {\em hetnet cell of $X^*$};
% i.e.  of the  macro BS  $X^*$ that  is closest to the origin, 
cf Figure~\ref{fig:hetnet-zero-macro-cel}~right.
Note, it is not necessarily the hetnet cell covering the origin. This is the region where all SUs receive the service from the 
macro BS $X^*$ serving the typical MU present at the origin.
Let $\overline N_{SU,macro}:=\E^0_{MU}[| \mathcal{U}_{static}(V_{het}(X^*))|]$ be the mean number of SUs 
present in the  hetnet cell of $X^*$  under the Palm distribution for
MUs. Note that, by the independence of 
$\mathcal{U}_{mobile}$, $\mathcal{U}_{static}$, $\Phi_{macro}$ and $\Phi_{micro}$, 
we can replace $\E^0_{MU}$ here simply by $\E$.
 
We have the following results regarding the two mean number of users.

\begin{theorem}\label{theorem:mean-number-of-MU-served-by-a-macro-BS}
 The mean number of MUs served by the macro BS serving a typical MU located at the origin is given by: 
 $\overline{N}_{MU,macro}=1+\bigg( \frac{1.2802 \lambda_L}{\lambda_{BS}p} + \frac{3.216 }{\pi \sqrt{\lambda_{BS}p}} \bigg) \lambda_{MU}$.
\end{theorem}
\begin{proof}
See Appendix~\ref{appendix}, 
Subsection~\ref{subsection:proof-of-mean-number-of-MU-theorem}. 
\end{proof}
For  any two points $(x_1,x_2)$ and $(y_1,y_2)$ on $\mathbb{R}^2$, we define 
$r_0:=(\frac{P_{micro}}{P_{macro}})^{\frac{1}{\beta}}\sqrt{y_1^2+y_2^2}$
and
$\mathcal{A}((x_1,x_2),(y_1,y_2))$ as the area of the union of two circles 
 with centers at $(x_1,x_2)$ and $(y_1,y_2)$ and radii
 $\sqrt{x_1^2+x_2^2}$ and $\sqrt{y_1^2+y_2^2}$ respectively.  
\begin{theorem}\label{theorem:mean-number-of-SU-served-by-a-macro-BS}
The mean number of SUs served by the macro BS serving a typical MU located at the origin is given by:  
\begin{eqnarray*}
&& \overline{N}_{SU, macro}\\ 
 &=& \lambda_{SU} \lambda_{BS} p \int_{(x_1,x_2) \in \mathbb{R}^2}\int_{(y_1,y_2) \in \mathbb{R}^2} \\
 && e^{-\lambda_{BS}p\mathcal{A}((x_1,x_2),(y_1,y_2)) - \lambda_{BS}(1-p) \pi r_0^2} dx_1 dx_2 dy_1 dy_2
\end{eqnarray*}
\end{theorem}
\begin{proof}
 See Appendix~\ref{appendix}, Subsection~\ref{subsection:proof-of-mean-number-of-SU-served-by-a-macro-BS}.
\end{proof}

The two mean number of users allow us to define the {\em mean MU
  throughput}  as
\begin{equation}
r_{MU}:=\frac{(1-\lambda_c vT_h)\E[R_{macro}(0)]}{\overline{N}_{SU,
    macro}+\overline{N}_{MU,macro}}\,.\label{e.r-MU}
\end{equation}

\subsubsection{An approximation for $\overline{N}_{SU,macro}$}
\label{subsubsection:approximation-for-mean-number-of-static-users-co-served-with-the-typical-mobile-user}
Since the expression for $\overline{N}_{SU,macro}$ in Theorem~\ref{theorem:mean-number-of-SU-served-by-a-macro-BS} 
is not easy for 
numerical computation, we approximate it by the expected number of static users served by a typical macro BS in the 
heterogeneous network $\lambda_{SU} \E_{macro}^0 [|V_{het}(0)|]$, where 
$\E _{macro}^0$ denotes expectation w.r.t. the Palm probability distribution  
given that a macro BS is located at the origin. 
We denote this approximation by $\hat{N}_{SU,macro}$.

\begin{theorem}\label{theorem:approximation-of-static-users-served-by-macro-BS}
 $\hat{N}_{SU,macro}:=\lambda_{SU} \E_{macro}^0 [|V_{het}(0)|]=\frac{pP_{macro}^{2/\beta}}{pP_{macro}^{2/\beta}+(1-p)P_{micro}^{2/\beta}} \times \frac{\lambda_{SU}}{\lambda_{BS}p}$.
\end{theorem}
\begin{proof}
 See Appendix~\ref{appendix}, 
 Subsection~\ref{subsection:proof-of-approximation-of-static-users-served-by-macro-BS}.
\end{proof}

\begin{theorem}\label{theorem:lower-bound-of-static-users-served-by-macro-BS}
 $\hat{N}_{SU,macro} \leq \overline{N}_{SU,macro}$.
\end{theorem}
\begin{proof}
 See Appendix~\ref{appendix}, 
 Subsection~\ref{subsection:lower-bound-of-static-users-served-by-macro-BS}.
\end{proof}

\subsubsection{An upper bound for $\overline{N}_{SU,macro}$}
\label{subsubsection:upper-bound-for-mean-number-of-static-users-co-served-with-the-typical-mobile-user}
\begin{theorem}\label{theorem:upper-bound-of-static-users-served-by-macro-BS}
 $\overline{N}_{SU,macro} \leq \frac{1.2802 \lambda_{SU}}{\lambda_{BS}p}$.
\end{theorem}
\begin{proof}
See Appendix~\ref{appendix}, Subsection~\ref{subsection:proof-of-upper-bound-of-static-users-served-by-macro-BS}.
\end{proof}

\subsection{Downlink Throughput of Static Users}
\label{ss.SU-throughput}
Following the same line of thought as for MUs, we 
denote by
  $$\E[R_{het}(0)]=\E[\log_2(1+\text{SIR}_{het}(0))]$$
 the {\em mean downlink bit-rate} at
  the origin from the base station whose hetnet cell is serving the
  origin;
$$\text{SIR}_{het}(0):=\frac{\max_iP_{i}(A|X_i|)^{-\beta}}{
\sum_{X_i\in\Phi}P_i(A|X_i|)^{-\beta}-\max_iP_{i}(A|X_i|)^{-\beta}}
\,.$$
Here $P_i \in \{P_{micro},P_{macro}\}$ is the transmit power from the base station located at $X_i \in \Phi$. 
We consider $\E[R_{het}(0)]$ as the peak bit-rate of the typical  SU.

When the condition~\eqref{e.hetnet-equivalence} is satisfied, by  the
network equivalence principle, cf~\cite{netequivalence,bartek-keeler15sinr-process-poisson-networks-factorial-moment-measures}, 
$\E[R_{het}(0)]=\E[R_{equivalent}(0)]$ where  this latter expectation
corresponds to $\E[R_{macro}(0)]$ in the one-tier  network consisting
of only macro BS of intensity $\lambda_{BS}$ and using transmit power $P$. 
Thus $\E[R_{equivalent}(0)]$ can be evaluated using the expressions in
Theorem~\ref{lemma:ccdf-sinr-tau-less-than-1} with $p=1$ and 
$P_{micro}=P_{macro}=P$, cf. Remark~\ref{rem.SIR-equivalence}.

In order to account for the resource sharing let $V_{het}(0)$ 
be the {\em zero hetnet cell}, i.e., the cell
  of the (macro or micro) BS that serves a typical SU when present at
  the origin; cf Figure~\ref{fig:hetnet-zero-cell-macro}.  

Denote by $\overline{N}_{SU, het}:=\E^0_{SU}[|\mathcal{U}_{static}(V_{het}(0))|]$ ($|\cdot|$ denotes the 
cardinality of the set here) 
the mean number of SUs present in the zero hetnet cell under the Palm distribution for SUs. 

Let $\overline{N}_{MU, het}:=\E^0_{SU}[| \mathcal{U}_{mobile}(V_{macro}(0))| \ind(V_{het}(0)=V_{het}(X^*))]$ 
be the mean number of MUs present in the zero macro  cell
under the Palm distribution for SUs, provided the BS serving the hetnet cell covering
the origin is a macro BS. Note that these are mobile users sharing the
service with the typical SU  at the origin. Note that, by the independence of 
$\mathcal{U}_{mobile}$, $\mathcal{U}_{static}$, $\Phi_{macro}$ and $\Phi_{micro}$, 
we can replace $\E^0_{SU}$ here simply by $\E$.

In order to express these two mean numbers denote by 
 $\mathcal{B}((x_1,x_2),(y_1,y_2))$ the area of the union of two circles 
 with centers at $(x_1,x_2)$ and $(y_1,y_2)$ and radii $(\frac{P_{micro}}{P_{macro}})^{\frac{1}{\beta}}\sqrt{x_1^2+x_2^2}$ 
 and $(\frac{P_{micro}}{P_{macro}})^{\frac{1}{\beta}}\sqrt{y_1^2+y_2^2}$ respectively. The function 
 $\mathcal{D}((x_1,x_2),(y_1,y_2))$ is defined as the area of the union of two circles 
 with centers at $(x_1,x_2)$ and $(y_1,y_2)$ and radii $(\frac{P_{macro}}{P_{micro}})^{\frac{1}{\beta}}\sqrt{x_1^2+x_2^2}$ 
 and $(\frac{P_{macro}}{P_{micro}})^{\frac{1}{\beta}}\sqrt{y_1^2+y_2^2}$ respectively. 

\begin{theorem}\label{theorem:mean-number-of-SU-served-by-a-hetnet-BS}
The mean number of SUs served by the (macro or micro) BS serving a typical SU located at the origin is given by:
\footnotesize
\begin{eqnarray*}
&& \overline{N}_{SU, het}=1+ \\
&& \lambda_{SU} \lambda_{BS}  p   \int_{(x_1,x_2,y_1,y_2) \in \mathbb{R}^4}  e^{-\lambda_{BS}p\mathcal{A}((x_1,x_2),(y_1,y_2))} \nonumber\\
&&\hspace{3em}  \times e^{-\lambda_{BS}(1-p) \mathcal{B}((x_1,x_2),(y_1,y_2))  }  dx_1 dx_2 dy_1 dy_2 \nonumber\\
 && + \lambda_{SU} \lambda_{BS} (1-p) \int_{(x_1,x_2,y_1,y_2) \in \mathbb{R}^4} e^{-\lambda_{BS}(1-p)\mathcal{A}((x_1,x_2),(y_1,y_2))} \nonumber\\
 && \hspace{3em}\times e^{-\lambda_{BS}p\mathcal{D}((x_1,x_2),(y_1,y_2))}   dx_1 dx_2 dy_1 dy_2.
\end{eqnarray*}
\normalsize
 \end{theorem}
\begin{proof}
 See Appendix~\ref{appendix}, Subsection~\ref{subsection:proof-of-mean-number-of-SU-served-by-a-hetnet-BS}.
\end{proof}

\begin{theorem}\label{theorem:mean-number-of-MU-served-by-a-hetnet-BS}
The mean number of MUs served by the same BS serving a typical SU located at the origin is given by:
 \begin{eqnarray*}
\overline{N}_{MU,het} =\frac{\lambda_L \lambda_{MU}}{\lambda_{SU}} \overline{N}_{SU,macro}
 \end{eqnarray*}
\end{theorem}
\begin{proof}
This follows from the mass transport principle, since  $\lambda_{SU} \overline{N}_{MU,het} = \lambda_L \lambda_{MU} \overline{N}_{SU,macro}$.
\end{proof}

We define the {\em mean SU  throughput} as
\begin{equation}\label{e.r-SU}
r_{SU}:=\frac{\E[R_{equivalent}(0)]}{\overline{N}_{SU,
    het}+\overline{N}_{MU,het}}\,.
\end{equation}

\subsubsection{An approximation for $\overline{N}_{MU, het}$}
\label{subsubsection:simpler-expression}
In order to obtain a computationally simple expression, we define 
$\hat{N}_{MU,het}:=\Pro \{ V_{het}(0)=V_{het}(X^*) \} \E [|\mathcal{U}_{mobile}(V_{macro}(0))|]$. This is an approximation 
to $\overline{N}_{MU,het}$ since the event $\{ V_{het}(0)=V_{het}(X^*) \}$ and the random variable 
$|\mathcal{U}_{mobile}(V_{macro}(0))|$ are not independent. 

\begin{theorem}\label{theorem:approximation-of-mobile-users-served-by-hetnet-BS}
 $\hat{N}_{MU,het}=\frac{pP_{macro}^{2/\beta}}{pP_{macro}^{2/\beta}+(1-p)P_{micro}^{2/\beta}} \times \frac{1.2802 \lambda_L \lambda_{MU}}{\lambda_{BS}p}$.
\end{theorem}
\begin{proof}
 See Appendix~\ref{appendix}, 
 Subsection~\ref{subsection:proof-of-approximation-of-mobile-users-served-by-hetnet-BS}.
\end{proof}

\subsubsection{An approximation for $\overline{N}_{SU, het}$}
\label{subsubsection:approximation-for-mean-number-of-static-users-co-served-with-a-static-user}
As an approximation to $\overline{N}_{SU, het}$, we define 
$\hat{N}_{SU,het}:=1+\Pro \{ V_{het}(0)=V_{het}(X^*) \} \E_{macro}^0 [|\mathcal{U}_{static}(V_{het}(0))|]+
\Pro \{ V_{het}(0) \neq V_{het}(X^*) \} \E_{micro}^0 [|\mathcal{U}_{static}(V_{het}(0))|]$. 

\begin{theorem}\label{theorem:approximation-of-static-users-served-by-hetnet-BS}
\footnotesize
\begin{eqnarray}
 \hat{N}_{SU, het} &=& 1+ \bigg( \frac{pP_{macro}^{2/\beta}}{pP_{macro}^{2/\beta}+(1-p)P_{micro}^{2/\beta}} \bigg)^2 \frac{\lambda_{SU}}{\lambda_{BS}p} \nonumber\\
 &&+ \bigg( \frac{(1-p)P_{micro}^{2/\beta}}{pP_{macro}^{2/\beta}+(1-p)P_{micro}^{2/\beta}} \bigg)^2 \frac{\lambda_{SU}}{\lambda_{BS}(1-p)} 
 \label{eqn:approximation-for-mean-number-of-static-users-co-served-with-a-static-user}
\end{eqnarray}
\end{theorem}
\normalsize
\begin{proof}
 See Appendix~\ref{appendix}, 
 Subsection~\ref{subsection:proof-of-approximation-of-static-users-served-by-hetnet-BS}.
\end{proof}

\begin{figure}[t!]
\begin{center}
\begin{minipage}{1\linewidth}
\begin{center}
\centerline{\includegraphics[width=0.4\linewidth]{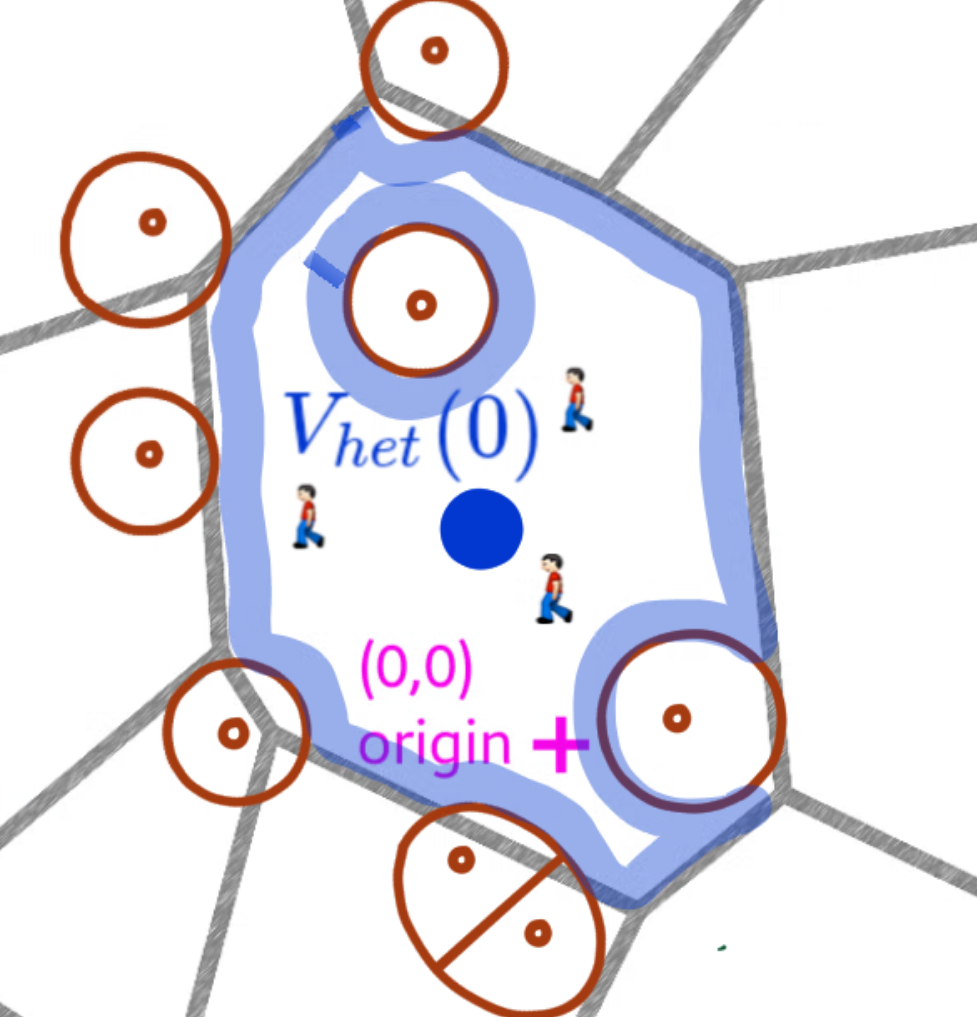}
\hspace{0.03\linewidth}
\includegraphics[width=0.33\linewidth]{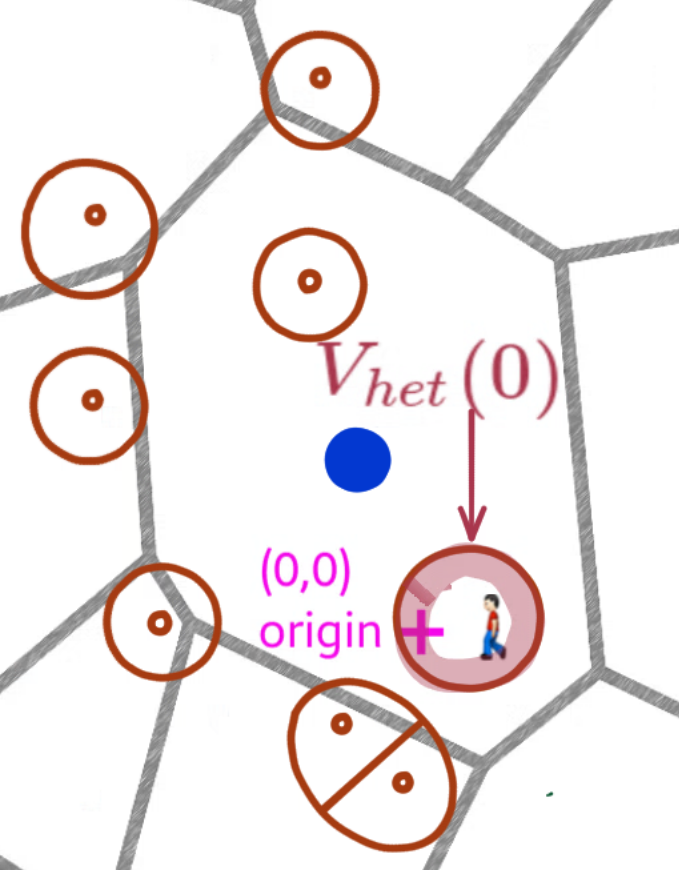}}
\vspace{-2ex}
\caption[Zero hetnet cell]{Zero hetnet cell; it may be a macro cell (left) or a micro cell (right).
It serves a typical SU at the origin.
\label{fig:hetnet-zero-cell-macro}}
\end{center}
\end{minipage}
\end{center}
\vspace{-2ex}
\end{figure}

\section{Optimal Design of the Heterogeneous Network}\label{section:optimal-design-of-the-network-using-stochastic-geometry}
Let us first consider the network in  which all stations transmit with the same power $P_{macro}=P_{micro}=P$. In this
homogeneous network, if the density of BSs $\lambda_{BS}$ is high (and
hence the cells are small) it might be advantageous to let the mobile users
connect only to some fraction $p<1$ of BSs to reduce the frequency
of handoffs during which  the connection is
not assured. Obviously too small  $p$ will result is low data rate 
due to large distance from BSs. Our model
allows us to calculate the optimal value of $p$ as a  function of mobile
speed $v$ and other network parameters. 
If a heterogeneous structure of the  network is allowed, it might be interesting to
further optimize jointly the performance of mobile and static  users
appropriately tuning the powers  $P_{macro}$ and $P_{micro}$. 

We consider first the optimization of the data rate of mobile users
and then optimize the throughput of static and mobile users.
The optimal proportion of macro stations and the transmit powers
provide a guideline for the design of the heterogeneous network.

\subsection{Optimizing Data Rate of Mobile Users}
\label{subsection:the-optimization-problem}
\label{subsubsection:optimization-problem-with-only-mobile-node-data-rate}
Consider the following  optimization of the data rate of mobile users
accounted for handoff events (cf Section~\ref{ss.rate-handoff}) 
within the class of equivalent
heterogeneous networks
%\begin{equation}
%\hspace{-1em}
%\begin{cases}
%\sup\limits_{0 \leq p \leq 1, P_{micro} \geq 0, P_{macro} \geq 0} \,(1-\lambda_cvT_h)\E[R_{macro}(0)]& \\[2ex]
%\text{given\ } p P_{macro}^{2/\beta}+(1-p)P_{micro}^{2/\beta}=P^{2/ \beta}&
%\end{cases}\label{eqn:single-mobile-user-modified-optimization-without-static-user-rate}
%\end{equation}
\begin{equation}
\sup\limits_{\scriptstyle 0 \leq p \leq 1,P_{micro}, P_{macro} \geq 0\atop
\scriptstyle\text{such that~\eqref{e.hetnet-equivalence} holds}} \,(1-\lambda_cvT_h)\E[R_{macro}(0)]
\label{eqn:single-mobile-user-modified-optimization-without-static-user-rate}
\end{equation}
for some given  $P$ and other model parameters.
The above problem needs to be solved numerically since the 
dependence of the integral~\eqref{e.ER} for $\E[R_{macro}(0)]$ with the 
distribution of SIR given in Theorem~\ref{lemma:ccdf-sinr-tau-less-than-1} cannot be
analytically evaluated with respect to the optimization
parameters. 

However, in order to have some insight into the structure of the MU rate
optimization let us
revisit~\eqref{eqn:single-mobile-user-modified-optimization-without-static-user-rate}
with $\E[R_{macro}(0)]$ approximated by
$\E[R_{macro}(0)\ind(\mathrm{SIR}_{macro}(0)>1)]$.~\footnote{This
  corresponds to the bit-rate with adaptive coding
  available only for SIR larger than~1.} 
Define the constant
$$C:=(\Gamma(1+2/\beta)\Gamma(1-2/\beta))^{-1}\int_{1}^\infty (2^t-1)^{-2/\beta}\,dt$$
Then, under constraint~\eqref{e.hetnet-equivalence}, 
$\E[R_{macro}(0)\ind(\mathrm{SIR}_{macro}(0)>1)]=Cp(P_{macro}/P_{micor})^{2/\beta}$
and hence
\begin{align}
&(1-\lambda_cvT_h)\E[R_{macro}(0) \ind(\mathrm{SIR}_{macro}(0)>1)]=\nonumber\\
&(1-4vT_h\sqrt{\lambda_{BS}p}/\pi)Cp(P_{macro}/P_{micro})^{2/\beta}.
\label{e.ER1SIR>1}
\end{align}
It is easy to see (calculating the derivative in $p$) that the value of~\eqref{e.ER1SIR>1} is maximized with $p\in[0,1]$
for $p=p^*$, where
\begin{equation}
p^*:=%(4/9)*\pi^2/(16v^2T_h^2\lambda_{BS})=
\min(1,\pi^2/(36v^2T_h^2\lambda_{BS}))\,.\label{e.p*}
\end{equation}

\begin{remark}\label{rem.MU-rate-opt}
\begin{enumerate}
\item Note that the value of $p^*$ in~\eqref{e.p*} does not depend on
  the power values $P_{macro}, P_{micro}$. In case of a homogeneous
  network these powers are fixed and equal to
  $P_{macro}=P_{micro}=P$. In this case $p^*$ can be interpreted as
  the optimal  fraction of BSs to which MUs should connect so as to
optimize their bit-rate. However, it has to be kept in mind that this formula is being used only to 
provide an intuitive explanation for not using all the base stations to 
serve the MUs. 
\item When a heterogeneous architecture is allowed,  the value
  of~\eqref{e.ER1SIR>1} with $p=p^*$ can be further maximized in
  $P_{macro}$ under constraint~\eqref{e.hetnet-equivalence}. I is easy to see that the optimal choice consists in
  taking $P_{macro}:=(p^*)^{\beta/2}P$ and $P_{micro}=0$. 
This means that using micro BSs is counterproductive from the point of
view of the maximization of the  bit-rate of MUs. Shutting down micro
BSs  and increasing appropriately the power of macro BSs (so as to
ensure the equivalent service for static users) appears to be an
optimal solution. This observation complies with the fact that 
the micro BSs are meant to provide extra capacity (and not rate-coverage)
to the network. Indeed, we shall see in the next section that 
only a joint  optimization of the throughput (which is a capacity metric) of static and mobile users
suggest a usage of micro BSs.
\end{enumerate}
\end{remark}

\subsection{Optimizing User Throughput}
\label{subsection:incorporating-data-rate-of-static-users-in-optimization}
We consider now optimization of the user throughput.
Our first observation is that if one focuses only on the throughput of
mobile users $r_{MU}$ given by~\eqref{e.r-MU}, i.e. considers 
\begin{equation}
\sup\limits_{\scriptstyle 0 \leq p \leq 1,P_{micro}, P_{macro} \geq 0\atop
\scriptstyle\text{such that~\eqref{e.hetnet-equivalence} holds}} \,
\frac{(1-\lambda_c vT_h)\E[R_{macro}(0)]}{\overline{N}_{SU,
    macro}+\overline{N}_{MU,macro}}\,,
\label{eqn:bandwidth-sharing-case-modified-optimization-without-static-user-rate}
\end{equation}
then, as in the case of MU rate optimization considered in
Section~\ref{subsubsection:optimization-problem-with-only-mobile-node-data-rate}, 
 the optimal solutions
consists in taking some $p<1$ when $vT_h\sqrt{\lambda_{BS}p}$ too large,
do not use micro BSs ($P_{micro}=0$) and adapt appropriately the
power of macro BSs (observed numerically). 
This can be again explained by the  observation that 
micro BSs are meant to provide capacity to static users. When 
$r_{SU}$ is absent from the optimization then there is no reason to use 
micro BSs.

This takes us to our ultimate problem of a joint optimization of
the throughput of static and mobile users
\begin{equation}
\sup\limits_{\scriptstyle 0 \leq p \leq 1,P_{micro}, P_{macro} \geq 0\atop
\scriptstyle\text{such that~\eqref{e.hetnet-equivalence} holds}} \,
r_{MU}+\xi r_{SU}
\label{eqn:optimization-problem-linear-combination}
\end{equation}
where $r_{MU}$ and $r_{SU}$ is the throughput of mobile and static
user given by~\eqref{e.r-MU} and \eqref{e.r-SU}, respectively, and 
$\xi$ is a multiplier that  captures the emphasis we put 
on the rate of the typical static user in the objective
function.~\footnote{E.g. taking the ratio of the intensities of the
  two types of users $\xi=\lambda_{SU}/(\lambda_{MU}\lambda_L)$ one
  considers in~\eqref{eqn:optimization-problem-linear-combination}
the mean throughput of the typical user (static or mobile).}
As we shall see in
Section~\ref{subsection:numerical-work-on-stochastic-geometry-results}, 
if enough emphasis is put on the throughput of static users then the
usage of micro stations is advantageous.

Let us denote the optimal solution of \eqref{eqn:optimization-problem-linear-combination} by 
$p^*(\xi)$, $P_{micro}^*(\xi)$ and $P_{macro}^*(\xi)$, and the corresponding optimal rates by 
$r_{MU}^*(\xi)$ and $r_{SU}^*(\xi)$.

\begin{lemma}\label{lemma:linear-combination-convex-increasing-in-xi}
 $r_{MU}^*(\xi)+\xi r_{SU}^*(\xi)$ is convex, increasing in $\xi$.
\end{lemma}
\begin{proof}
 See Appendix~\ref{appendix}, 
 Subsection~\ref{subsection:proof-of-linear-combination-convex-increasing-in-xi}.
\end{proof}

\begin{lemma}\label{lemma:linear-combination-optimization-rate-of-mobile-user-decreases-in-xi}
 $r_{SU}^*(\xi)$ is increasing in $\xi$, and $r_{MU}^*(\xi)$ is decreasing in $\xi$.
\end{lemma}
\begin{proof}
 See Appendix~\ref{appendix}, 
 Subsection~\ref{subsection:proof-of-linear-combination-optimization-rate-of-mobile-user-decreases-in-xi}.
\end{proof}

Problem~\eqref{eqn:optimization-problem-linear-combination} can be used to solve the following constrained problem:
\begin{eqnarray}
&& \sup\limits_{\scriptstyle 0 \leq p \leq 1,P_{micro}, P_{macro} \geq 0\atop
\scriptstyle\text{such that~\eqref{e.hetnet-equivalence} holds}} \,
r_{MU} \nonumber\\
&& \text{such that } r_{SU} \geq r_0
\label{eqn:constrained-optimization-problem}
\end{eqnarray}

The following standard result tells us how to choose $\xi$.
\begin{theorem}\label{theorem:constrained-and-unconstrained}
 If there exists $\xi^* \geq 0$ such that, under the optimal solution of \eqref{eqn:optimization-problem-linear-combination} 
 with $\xi=\xi^*$, the constraint in \eqref{eqn:constrained-optimization-problem} is met with equality, then that solution 
 is optimal for the constrained problem \eqref{eqn:constrained-optimization-problem} as well.
\end{theorem}

\subsection{Numerical Results and Insights to the Network Design Problem}\label{subsection:numerical-work-on-stochastic-geometry-results}

\begin{figure}[!t]
\begin{center}
\includegraphics[height=7cm, width=9cm]{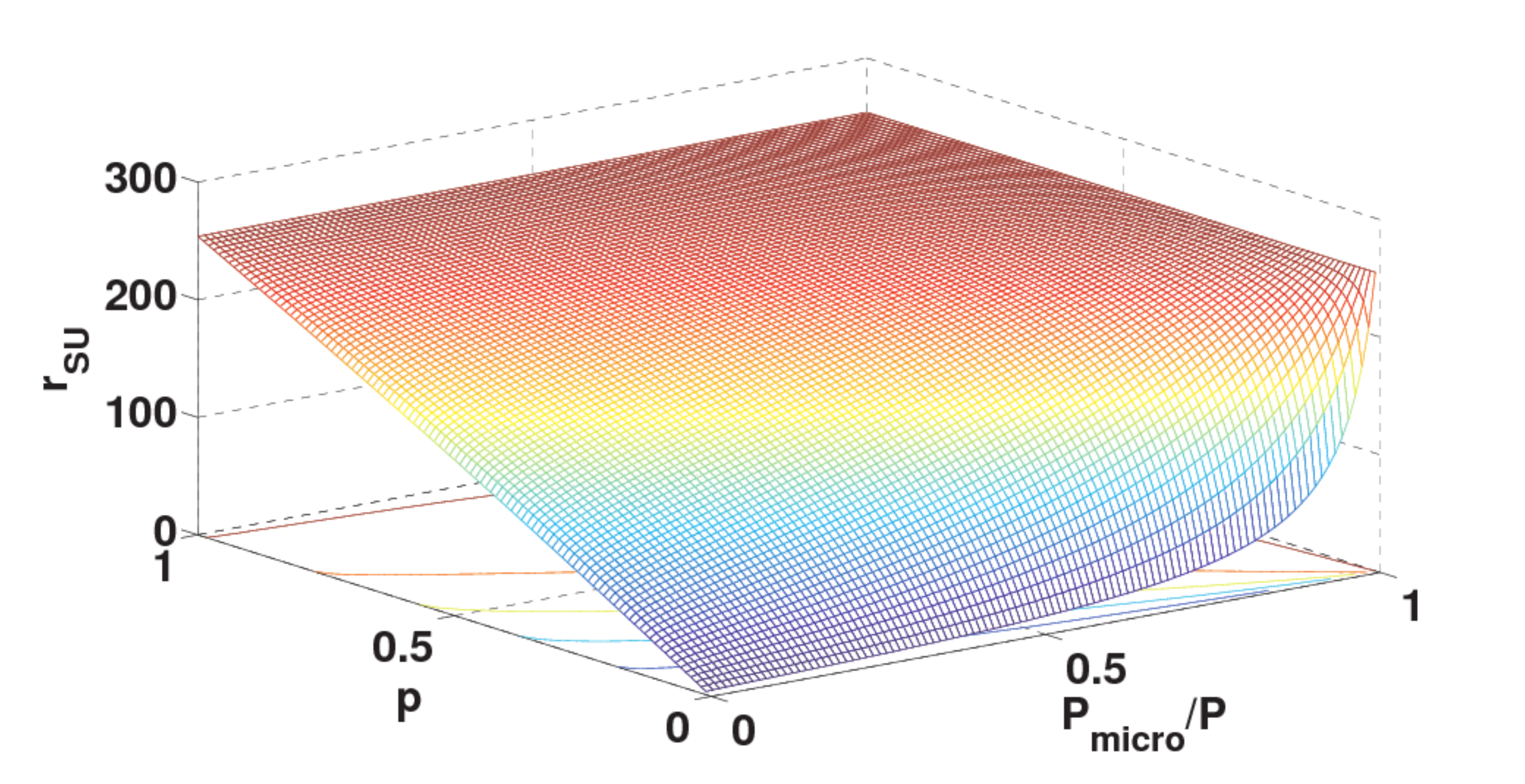}
\end{center}
\vspace{-4mm}
\caption{Variation of $r_{SU}$ with $p$ and $\frac{P_{micro}}{P}$ for system parameters chosen in 
Section~\ref{subsection:numerical-work-on-stochastic-geometry-results}. When $p$ and $P_{micro}$ are both $0$, 
there are only micro base stations with $0$ transmit power. Hence, all users get 
 zero data rate. This situation is excluded from this plot.}
\label{fig:rsu}
\vspace{-3mm}
\end{figure}

\begin{figure}[!t]
\begin{center}
\includegraphics[height=7cm, width=9cm]{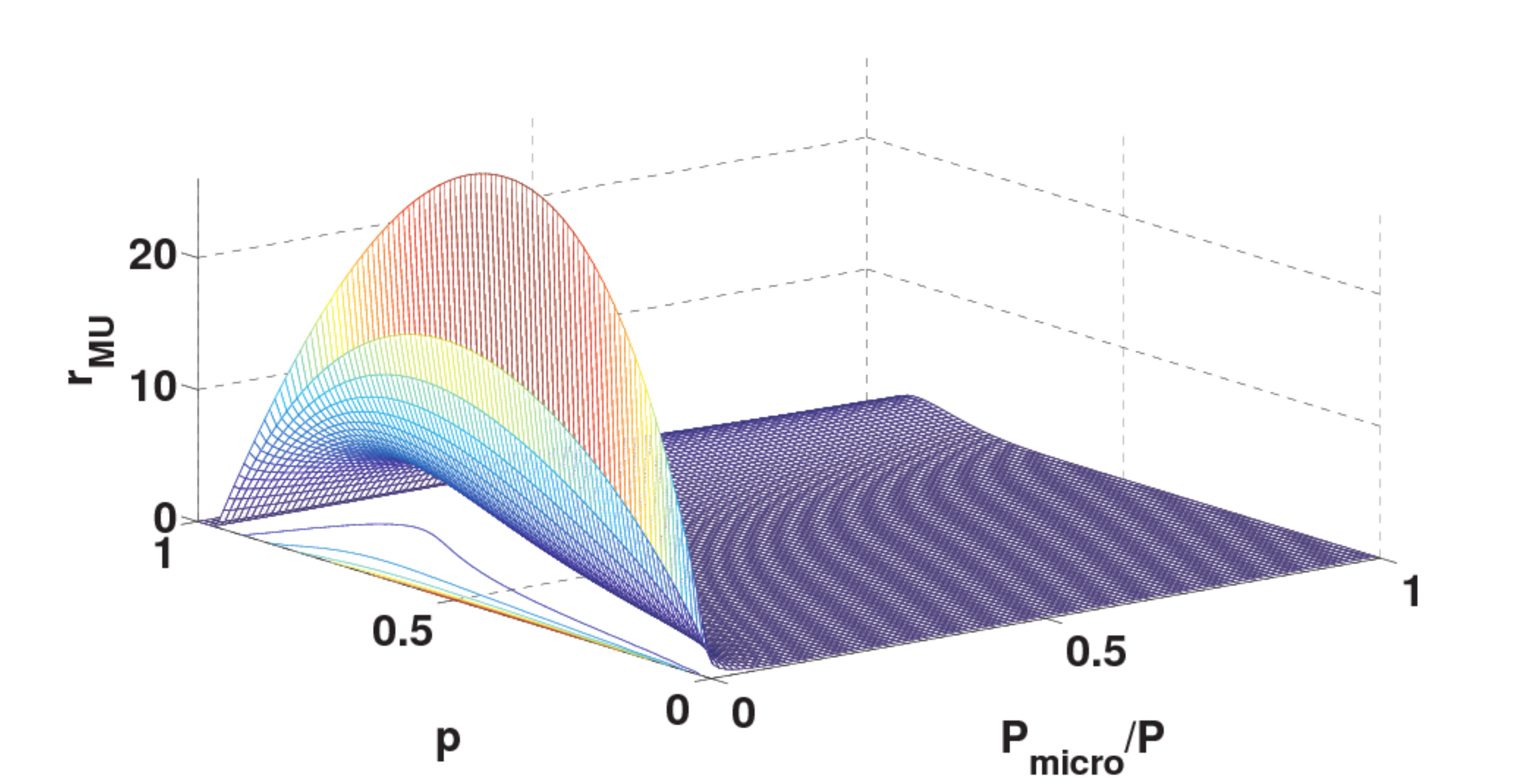}
\end{center}
\vspace{-4mm}
\caption{Variation of $r_{MU}$ with $p$ and $\frac{P_{micro}}{P}$ for system parameters chosen in 
Section~\ref{subsection:numerical-work-on-stochastic-geometry-results}. The situation when $p$ and $P_{micro}$ are both $0$ is 
not included in the plot.}
\label{fig:rmu}
\end{figure}

We consider a networks with $P=1$~unit (the results presented in this section are invariant w.r.t. $P$; the reason 
is that if we scale the transmit power at each BS by a constant factor, the SIR at any location and the 
cell structure remain unchanged), $\lambda_{BS}=\frac{1}{2500}/\text{m}^2$ (one base station per 
$50 \text{m} \times 50 \text{m}$ area), $v=20 \text{m}/\text{sec}$ ($72$~kmph), 
$T_h=2$~second (note that, the results in this section will remain unchanged if we vary $v$ and 
$T_h$ while keeping their product constant), 
$\lambda_{SU}=\frac{1}{400}/\text{m}^2$ (one static user per $20 \text{m} \times 20 \text{m}$ area), 
$\lambda_{MU}=\frac{1}{20}/ \text{m}$ (one MU per $20$~m distance) and $\lambda_L=\frac{\sqrt{2}}{50}/ \text{m}$ 
(equivalent to the length of the 
diagonal in a $50 \text{m} \times 50 \text{m}$ square).

The variation of $r_{SU}$ and $r_{MU}$ with $p$ and $\frac{P_{micro}}{P}$ are shown in 
Figure~\ref{fig:rsu} and \ref{fig:rmu} respectively. Several observations can be made 
from these plots as discussed below:
\begin{itemize}
\item As $p$ or $P_{micro}$ increases, the network becomes more and more homogeneous, and $r_{SU}$ increases.
\item As $P_{micro}$ increases for a fixed $p$, the throughput $r_{MU}$ decreases because $P_{macro}$ decreases and 
interference from micro BSs increases.
\item As $p$ increases keeping $P_{micro}$ fixed, $r_{MU}$ first increases and then decreases. Initially $r_{MU}$ 
increases because more macro BSs are added that can serve the MUs. But beyond certain value of $p$, the throughput loss 
due to frequent handoff starts dominating, and hence $r_{MU}$ decreases with $p$.
\item The performance of MUs is very sensitive to the value of $P_{micro}$; $r_{MU}$ drops rapidly as $P_{micro}$ 
increases from $0$. Hence, the value of $P_{micro}$ should be small if we want high $r_{MU}$, and the value of $p$ has to be 
chosen optimally as shown in Figure~\ref{fig:rsu} and \ref{fig:rmu}.
\end{itemize}

\begin{table}[t!]
\centering
\footnotesize
\begin{tabular}{|c |c |c|c|c|c|}
\hline
     $\xi$            & $p^*$ & $\frac{P_{micro}^*}{P}$ & $\frac{P_{macro}^*}{P}$ & $r_{SU}^*(\xi)$  & $r_{MU}^*(\xi)$   \\ 
                      &       &                       &                       & (bits/sec/Hz)  &   (bits/sec/Hz)                    \\ \hline
0.001    & 0.4 & 0 & 4.9704  & 107.0036  & 33.2543 \\ \hline
0.01     & 0.41 & 0 & 4.7602 & 109.5806 & 33.2408 \\ \hline
0.1     & 0.5 & 0 & 3.3636 & 132.5683 & 31.9570 \\ \hline
0.2     & 0.62 & 0 & 2.3084 & 162.6537 &  27.3609 \\ \hline
0.29        &   0.73    &  0  &  1.7345  &  189.6808  &  20.6754  \\ \hline
0.3     &  0.9 & 1 & 1 & 253.8792 & 1.4436 \\ \hline
0.4     &  0.9 & 1 & 1 & 253.8792 & 1.4436 \\ \hline
0.5     &  0.9 & 1 & 1 & 253.8792 & 1.4436 \\ \hline
1     &  0.9 & 1 & 1 & 253.8792 & 1.4436 \\ \hline
 \end{tabular}
 \normalsize
\caption{Effect of $\xi$ on the optimal heterogeneous network design under the equivalent network condition}
\label{table:unbounded-and-bounded-power}
\end{table}

Now we focus on numerical solution to the problem \eqref{eqn:optimization-problem-linear-combination}. 
Let us recall that the optimal solution of \eqref{eqn:optimization-problem-linear-combination} is denoted  by 
$p^*(\xi)$, $P_{micro}^*(\xi)$ and $P_{macro}^*(\xi)$, and the corresponding optimal rates are denoted by 
$r_{MU}^*(\xi)$ and $r_{SU}^*(\xi)$. 
From the numerical results in Table~\ref{table:unbounded-and-bounded-power}, we observe that $p^*<1$, $P_{macro}^*>P$ and 
$P_{micro}^*=0$ for small values of $\xi$, whereas $p^*<1$, $P_{macro}^*=P_{micro}^*=P$ 
above certain value of $\xi$.\footnote{We have used programs from  \cite{keeler-programs} to compute the function 
$\mathcal{J}_{n, \beta}(\cdot)$.} This is explained 
by the fact that small $\xi$ puts more weightage on the throughput of mobile users, and hence micro base stations 
(which cause interference to MU downlink) are shut down, and a fraction of BSs are used as macro BSs with high power. On the other hand, large value 
of $\xi$ puts more weightage on the throughput of static users, thereby resulting in a homogeneous network design with 
$P_{macro}^*=P_{micro}^*=P$. The results also demonstrate that using macro BSs can improve the rate of MUs in practice 
(which is not intuitive since macro BSs reduce handoff rate, but at the same time a typical MU is co-served with more SUs and MUs, and 
the macro cell size increases). The interesting part of the observation is 
that the network should always be homogeneous in both cases; this is a consequence of the sensitivity of $r_{MU}$ with 
$p$ and $P_{micro}$ as discussed in the previous paragraph. 
Of course, the optimal design 
will depend on parameters such as $\lambda_{BS}$, $v$, $T_h$, $\lambda_{MU}$ and $\lambda_L$, and also the choice of $\xi$; hence, 
choice of the optimal fraction and power levels of macro and micro base stations will depend on the estimates of user densities and user velocity 
estimates. The choice will also depend on the physicals constraints of the system designer (e.g., availability and cost for macro and micro 
BSs, maximum transmit power available at macro BSs etc.) For example, for large $v$, the macro BSs may need very high power, but the 
commercially available BSs may not be able to meet this power requirement.

From Table~\ref{table:unbounded-and-bounded-power}, we can solve the constrained problem 
\eqref{eqn:constrained-optimization-problem} by choosing appropriate $\xi^*$ as 
described in Theorem~\ref{theorem:constrained-and-unconstrained}.

% However, the optimal solution of this problem may not be able to yield a desired feasible combination of 
% $r_{SU}$ and $r_{MU}$; in such cases, the network designer should skip the optimization problem~\eqref{eqn:optimization-problem-linear-combination} 
% and choose the values of $p$ and $P_{micro}$ that can yield a target rate for static and mobile users. For example, in the above example, if 
% a target $r_{SU}=229.5962$ and $r_{MU}=4.7566$ is set, then the network designer has to choose $p^*=0.8$, 
% $P_{micro}=0.2$ and $P_{macro}=1.2777$. 

\section{Conclusion}\label{section:conclusion}
In this paper, we have explored the design (or cell planning) of heterogeneous cellular networks to combat throughput 
loss due to handoff. Analytical results and numerical exploration demonstrate the performance and tradeoffs.

Even though we have solved the basic problem in this paper, there are many possible extensions as well as 
numerous issues to improve upon: 
(i) We assumed full interference from all base stations, but it would be of interest to consider the effect of frequency reuse 
(e.g., in wireless standards such as LTE-A) on the design and resource allocation problems addressed in this paper, (ii) An interesting problem will be 
cell planning for other models of user mobility such as random waypoint model (\cite{hyytia-virtamo13random-waypoint-model}), 
(iii) In practice, there can be multiple possible values of user velocity 
(generally the network operator will classify user velocities into a discrete set). 
Hence, a multi-tier network architecture needs to be developed. However, our numerical work has left the open question about the choice of 
design parameters in such multi-tier networks, since the numerical work with formulation~\eqref{eqn:optimization-problem-linear-combination} 
proposes a homogeneous network whereas that formulation cannot achieve all feasible tuples of $(r_{MU},r_{SU})$ (and 
a solution {\em suboptimal} to this formulation has to be adopted),  
(iv) Extension to the realistic situation where (random) shadowing variation over space modulates the path-loss function is a challenging problem, since 
this will result in unpredictable behaviour of handoff request generation process due to lack of 
an accurate statistical characterization of the variation of shadowing over space, 
(v) Cooperative transmission by multiple base stations to the mobile users can also be explored as 
a potential solution for throughput loss due to handoff, 
(vi) Extension to future 5G network models is very challenging, since the association of the mobile users to the densely 
deployed base stations are supposed to change rapidly over time, resulting in an unprecedented amount of handoff traffic, 
(vii) In this paper, we have assumed that macro base stations serve SUs and MUs, while micro BSs serve only SUs. There can be other 
service disciplines, such as macro base stations serving only MUs and micro base stations serving only SUs, or a fraction of macro 
BSs serving only MUs and some other fraction serving SUs and MUs; performance of such 
service disciplines needs to be investigated. (viii) The results in this paper 
do not guarantee a minimum throughput for the users all the time; if a user is located where there is no BS close to it, it will experience poor 
throughput. A grid-like base station process may be able to solve this issue, but the optimal 
design procedure of such networks needs to be explored. (ix) A simple way to address the problem of handoff failure due 
to overload in the target cell would be to multiply the throughput of MUs  by the probability that the target 
macro BS rejects 
a handoff request (this probability has to be averaged over all macro BSs), 
and solve the same optimization  as in this paper. However, in practice, this probability will 
be a function of the network design parameters and user densities; hence, the numerical optimization will be more complicated 
than that solved in this paper. We propose 
to pursue some of these topics in our future research endeavours.

{\small
\bibliographystyle{unsrt}
\bibliography{arpan-techreport}
}

  \vspace{-10mm}

\begin{IEEEbiography}[{\includegraphics[width=1in,height=1in,clip,keepaspectratio]{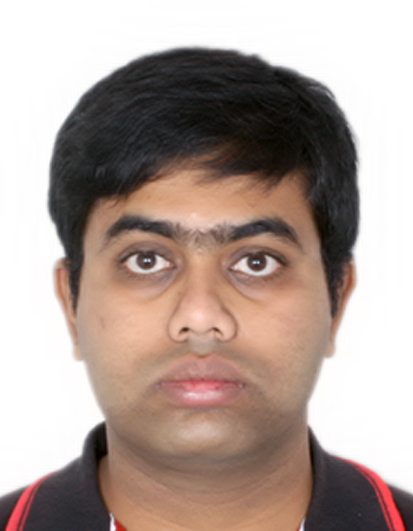}}]{Arpan 
Chattopadhyay} obtained his B.E. in Electronics and Telecommunication Engineering from Jadavpur University, 
Kolkata, India in the year 2008, and M.E. and Ph.D in Telecommunication Engineering from Indian Institute of Science, 
Bangalore, India in the year 2010 and 2015, respectively. He is currently working in INRIA, Paris as a postdoctoral researcher. 
His research interests include  networks, machine learning, information theory and control.
    \end{IEEEbiography}

        \vspace{-10mm}

\begin{IEEEbiography}[{\includegraphics[width=1in,height=1in,clip,keepaspectratio]{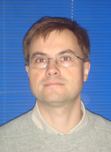}}]
{Bartlomiej Blaszczyszyn} received his PhD degree and Habilitation qualification
in applied mathematics from University of Wrocław (Poland) in 1995
and 2008, respectively. He is now a Senior Researcher at Inria (France), and
a member of the Computer Science Department of Ecole Normale Supérieure
in Paris. His professional interests are in applied probability, in particular in
stochastic modeling and performance evaluation of communication networks.
He coauthored several publications on this subject in major international
journals and conferences, as well as a two-volume book on {\em Stochastic
Geometry and Wireless Networks} NoW Publishers, jointly with F. Baccelli.
   
    \end{IEEEbiography}

    \vspace{-10mm}

   \begin{IEEEbiography}[{\includegraphics[width=1in,height=1in,clip,keepaspectratio]{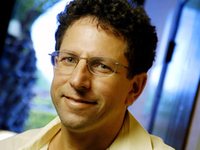}}]
   {Eitan Altman}  received the B.Sc. degree in electrical
engineering (1984), the B.A. degree in
physics (1984) and the Ph.D. degree in electrical
engineering (1990), all from the Technion-Israel
Institute, Haifa. In (1990) he further received his
B.Mus. degree in music composition in Tel-Aviv
University. Since 1990, he has been with INRIA
(National research institute in informatics and
control) in Sophia-Antipolis, France. His current
research interests include performance evaluation
and control of telecommunication networks and
in particular congestion control, wireless communications and networking
games. He is in the editorial board of several scientific journals: JEDC,
COMNET, DEDS and WICON. He has been the (co)chairman of the
program committee of several international conferences and workshops
on game theory, networking games and mobile networks.
   \end{IEEEbiography}

% \begin{figure*}[t!]
% \centering{\huge{\bf Supplementary Material} }
% \end{figure*}

\renewcommand{\thesubsection}{\Alph{subsection}}

\appendices

 \section{}\label{appendix}
 
 \subsection{Expression for $\mathcal{J}_{n,\beta}(\cdot)$ function}
 \label{subsection:expression-for-Jn}
 
 The function $\mathcal{J}_{n,\beta}(x_1,x_2,\cdots,x_n)$ is defined as follows:
 
 \footnotesize
 \begin{eqnarray*}
 \frac{(1+\sum_{j=1}^n x_j)}{n}  \int_{[0,1]^{n-1}}  \frac{ \Pi_{i=1}^{n-1} v_i^{i(\frac{2}{\beta}+1)-1} (1-v_i)^{\frac{2}{\beta}}  }{ \Pi_{i=1}^n (x_i+\eta_i) }     dv_1 dv_2 \ldots dv_{n-1}
 \end{eqnarray*}
\normalsize
  where 
  \begin{eqnarray*}
  \eta_1 &=& v_1 v_2 \ldots v_{n-1}\\
  \eta_2 &=& (1-v_1) v_2 \ldots v_{n-1}\\
    \eta_3 &=& (1-v_2) v_3 \ldots v_{n-1}\\
&& \cdots \\
    \eta_n &=& 1-v_{n-1}\\
 \end{eqnarray*}

 The function $\mathcal{J}_{n,\beta}(x)$ is calculated by substituting $x_1=x_2=\cdots=x_n=x$ in the expression for 
 $\mathcal{J}_{n,\beta}(x_1,x_2,\cdots,x_n)$.

  \subsection{Proof of Theorem~\ref{lemma:ccdf-sinr-tau-less-than-1}}
\label{subsection:proof-of-lemma-ccdf-sinr-tau-less-than-1}
Denote $\text{SIR}=\text{SIR}_{macro}(0)$ and denote the interference at the origin 
by the base stations belonging to $\Phi_{micro}$ by $I_{micro}$. 
For a given realization of $I_{micro}$, using
\cite[Corollary~$19$]{bartek-keeler15sinr-process-poisson-networks-factorial-moment-measures}),
we can write 
 that, 
$\Pro (\text{SIR}>\tau|I_{micro})=\sum_{n=1}^{\lceil{\frac{1}{\tau}}\rceil}(-1)^{n-1} \tau_n^{-\frac{2n}{\beta}} \mathcal{J}_{n, \beta} ( \tau_n ) 
\mathcal{I}_{n,\beta}(I_{micro}a^{-\frac{\beta}{2}})$, where $\tau_n:=\frac{\tau}{1-(n-1)\tau}$, $a=\frac{ \pi \lambda_{BS} p P_{macro}^{\frac{2}{\beta}} }{A^2}$,
and the function 
$\mathcal{I}_{n,\beta}(\cdot)$ is defined in \cite[Equation~$(13)$]{bartek-keeler15sinr-process-poisson-networks-factorial-moment-measures}. 
For completeness, we provide the expression 
$\mathcal{I}_{n,\beta}(x):=\frac{2^n \int_0^{\infty} u^{2n-1}e^{-u^2-u^{\beta}x \Gamma(1-2 / \beta)^{-\beta/2}} du }{\beta^{n-1}(n-1)! \Gamma(1-2 / \beta)^n \Gamma(1+2 / \beta)^n}$. 

Unconditioning $\mathcal{I}_{n,\beta}(I_{micro}a^{-\frac{\beta}{2}})$
over $I_{micro}$, we obtain: 
\begin{align*}
&\E[\mathcal{I}_{n,\beta}(I_{micro}a^{-\frac{\beta}{2}})]\\
&=
\bigg( \beta^{n-1} (n-1)! \bigg)^{-1} \bigg( \frac{2}{\Gamma(1-\frac{2}{\beta}) \Gamma(1+\frac{2}{\beta})}
\bigg)^n \\
& \times \int_0^{\infty} u^{2n-1}e^{-u^2}\E[e^{-I_{micro}
  a^{-\frac{\beta}{2}}u^{\beta}
  \Gamma(1-\frac{2}{\beta})^{-\frac{\beta}{2}}  }] du.
\end{align*}
From
\cite[Equation~$(3.3)$]{bartek-muhlethaler15interference-aware-sinr-coverage-aloha-arxiv},
we can write the Laplace transform 
\begin{align*}
\E[e^{-I_{micro}  a^{-\frac{\beta}{2}}u^{\beta}
  \Gamma(1-\frac{2}{\beta})^{-\frac{\beta}{2}}}]&=\mathcal{L}_{I_{micro}}(a^{-\frac{\beta}{2}}u^{\beta} \Gamma(1-\frac{2}{\beta})^{-\frac{\beta}{2}})\\
&=e^{-\lambda_{BS} (1-p) A^{-2} a^{-1} u^2 \pi
  P_{micro}^{\frac{2}{\beta}} }\\
&=e^{- u^2 (1-p)/p (P_{micro}/P_{macro})^{\frac{2}{\beta}}}.
\end{align*}
%{\footnotesize
%\begin{eqnarray*}
%&& \Pro (\text{SIR}>\tau) \\
%&=& \sum_{n=1}^{\lceil{\frac{1}{\tau}}\rceil}(-1)^{n-1} \tau_n^{-\frac{2n}{\beta}} \mathcal{J}_{n, \beta} ( \tau_n ) \E  \mathcal{I}(I_{micro}a^{-%\frac{\beta}{2}}) \\
%&=& \sum_{n=1}^{\lceil{\frac{1}{\tau}}\rceil}(-1)^{n-1} \tau_n^{-\frac{2n}{\beta}} \mathcal{J}_{n, \beta} ( \tau_n ) \times \bigg( \beta^{n-1} (n-%1)! \bigg)^{-1} \\
%&& \times \bigg( \frac{2}{\Gamma(1-\frac{2}{\beta}) \Gamma(1+\frac{2}{\beta})} \bigg)^n \\
%&& \times \int_0^{\infty} u^{2n-1}e^{-u^2}e^{-I_{micro}  a^{-\frac{\beta}{2}}u^{\beta} \Gamma(1-\frac{2}{\beta})^{-\frac{\beta}{2}}  } du  \\
%&=& \sum_{n=1}^{\lceil{\frac{1}{\tau}}\rceil}(-1)^{n-1} \tau_n^{-\frac{2n}{\beta}} \mathcal{J}_{n, \beta} ( \tau_n ) \times \bigg( \beta^{n-1} (n-%1)! \bigg)^{-1} \\
%&& \times \bigg( \frac{2}{\Gamma(1-\frac{2}{\beta}) \Gamma(1+\frac{2}{\beta})} \bigg)^n \\
%&& \times \int_0^{\infty} u^{2n-1}e^{-u^2}\mathcal{L}_{I_{micro}}\bigg(a^{-\frac{\beta}{2}}u^{\beta} \Gamma(1-\frac{2}{\beta})^{-\frac{\beta}{2}} % \bigg) du  \\
%\end{eqnarray*}}
%Plugging in the value of $\mathcal{L}_{I_{micro}}(\cdot)$ and simplifying, we obtain:
Putting this 
into the expression of $\Pro (\text{SIR}>\tau)$ we obtain
\footnotesize
\begin{eqnarray*}
&& \Pro (\text{SIR}>\tau) \\
&=& \sum_{n=1}^{\lceil{\frac{1}{\tau}}\rceil}(-1)^{n-1} \tau_n^{-\frac{2n}{\beta}} \mathcal{J}_{n, \beta} ( \tau_n ) \times \bigg( \beta^{n-1} (n-1)! \bigg)^{-1} \\
&& \times \bigg( \frac{2}{\Gamma(1-\frac{2}{\beta}) \Gamma(1+\frac{2}{\beta})} \bigg)^n 
\times \int_0^{\infty} u^{2n-1}e^{-u^2 \bigg(  1+\frac{(1-p)P_{micro}^{\frac{2}{\beta}}}{p P_{macro}^{\frac{2}{\beta}}} \bigg) } du  \\
\end{eqnarray*}
\normalsize

 Substituting $v:=u^2 \bigg( 1+\frac{(1-p)P_{micro}^{\frac{2}{\beta}}}{p P_{macro}^{\frac{2}{\beta}}}  \bigg)$ in the above 
integral and simplifying it further, we prove the lemma.
It might be also useful to observe that $2/(\Gamma(1-\frac{2}{\beta})\Gamma(1+\frac{2}{\beta})
=\beta/\pi\sin(2\pi/\beta+\pi)$

 \subsection{Proof of Remark~\ref{rem:ergodicity-along-a-line-without-handoff-for-a-single-user}}
\label{subsection:proof-of-ergodicity-theorem}
 Consider a given  directed line $l$ on the plane.
 Note that the couple $(\Phi_{macro},\Phi_{micro})$ is ergodic (by \cite[Proposition~$2.6$]{meester-roy96continuum-percolation})  and the $\sigma$-field $\sigma(\Phi_{macro},\Phi_{micro})$ is countably generated.
% by events of the form $\{\omega_{macro}(A)=j, \omega_{micro}(B)=k \}$, with $j,k \in \{0,1,2,\cdots\}$; $\omega_{macro}(A)$ and 
% $\omega_{micro}(B)$ are the cardinalities of 
% $\Phi_{macro} \cap A$ and $\Phi_{micro} \cap B$,  where $A$ and $B$ are open rectangles in $\mathbb{R}^2$ with rational coordinates). 
 Hence,
 by~\cite[Proposition~$2.7$]{meester-roy96continuum-percolation}, the sample average data 
 rate (average of $R_{macro}(x)$  taken over all points $x$ along  line $l$), exists and is almost surely equal to a constant
 $c_l$ for all, except for at most countably many number of lines
 $l$, with $c_l=\E[R_{macro}(x_l)]$ for arbitrary   $x_l\in l$.
 Since the couple $(\Phi_{macro},\Phi_{micro})$ is translation invariant
 $\E[R_{macro}(x_l)]=\E[R_{macro}(0)]$. Finally, probability that the
 Poisson line process $\mathcal{L}$ places some of its lines
 in the at most countable subset of lines $l$ is null.

\subsection{Proof of Lemma~\ref{lemma:lower-bound-on-fraction-of-nonoutage-length-over-real-line}}
\label{subsection:proof-of-lower-bound-on-fraction-of-nonoutage-length-over-real-line}
% By ergodicity arguments similar to the proof of Theorem~\ref{theorem:ergodicity-along-a-line-without-handoff-for-a-single-user}, we can say that %the L.H.S is equal to 
%$\E |[0,1] \cap \mathcal{T}_{nohandoff}|$. The expected number of cell crossings in the interval $[0,1]$ is $\lambda_c$, and hence 
%$\E |[0,1] \cap \mathcal{T}_{handoff}| \leq \lambda_c vT_h$. Hence, 
%$\E |[0,1] \cap \mathcal{T}_{nohandoff}| \geq (1-\lambda_cvT_h)$.
Recall that $\lambda_c$ is the density of handoffs (macro cell boundary
crossings)  on every line of $\mathcal{L}$. Observe that $vT_h$ is the length of
the segment  corresponding to each handoff event. 
If the segments corresponding to
different handoffs on a given line  were disjoint the fraction of the
line where mobiles are not in handoff would be equal to
$1-\lambda_cvT_h$. The inequality results from the fact that 
two different handoff events may have overlapping segments on the
line, or, in other words, that (in case of crossing small cells) a MU may not recover from the previous
handoff before going into the next one.

 \subsection{Proof of Theorem~\ref{theorem:mean-number-of-MU-served-by-a-macro-BS}}
\label{subsection:proof-of-mean-number-of-MU-theorem}
Under $\E_{MU}^0$ the typical MU is located at the origin, while
other MUs form a Poisson process of intensity $\lambda_{MU}$ on lines
of the original, independent $\mathcal{L}$ appended with one extra line, crossing the origin and
independently, uniformly oriented. This is the line  along which
moves the typical user.  
Denote by $L_{macro}(0)$ the  intersection of this extra
line with $V_{macro}(0)$ and by  $|V_{macro}(0)|$ and $|L_{macro}(0)|$
the respective area and  length.
 Knowing that the expected total length of intersection
of $\mathcal{L}$ with any given set is equal to $\lambda_{L}$ times
the surface of this set we have: 
\begin{equation}\label{eqn:intermediate-expression-mean-number-of-mobile-users-served-simultaneously}
 \overline{N}_{MU,macro}=1+(\lambda_L \E[|V_{macro}(0)|+\E_{MU}^0[|L_{macro}(0)|])\lambda_{MU},
\end{equation}
where we replaced $\E_{MU}^0$[...] by $\E[...]$ in the first term due
to independence of $\mathcal{L}$ and $\Phi$.
By the inverse formula of Palm calculus, cf.~\cite[Theorem 4.1 and
  Corollary 4.4]{FnT1}, we have
$\E[|V_{macro}(0)|]=\frac{\E_{macro}^0[|V_{macro}(0)|^2]}{\E_{macro}^0[|V_{macro}(0)|]},$
where $\E_{macro}^0$ corresponds to Palm distribution for
$\Phi_{macro}$ (i.e. given a macro BS at the origin).
 Now, $\E_{macro}^0[|V_{macro}(0)|]=1/(\lambda_{BS}p)$, (cf ~\cite[Corollary 4.3]{FnT1})
and  the variance $\mathbf{Var}_{macro}^0(|V_{macro}(0)|)=\frac{0.2802}{(\lambda_{BS}p)^2}$
(cf \cite[Table~$9.5$]{chiu-etal13stochastic-geometry-and-its-applications})
Hence, $\E_{macro}^0[|V_{macro}(0)|^2]=\mathbf{Var}_{macro}^0(|V_{macro}(0)|)+(\E_{macro}^0[|V_{macro}(0)|]^2=\frac{1.2802}{(\lambda_{BS}p)^2}$
and consequently $\E[|V_{macro}(0)|]=\frac{1.2802}{\lambda_{BS}p}$.

Regarding the length of $L_{macro}(0)$ under $\E_{MU}^0$, we can observe that it
has the same distribution as the length 
of the interval $\tilde L(0)$ between two consecutive handoffs
(crossings of the macro cell boundary) of, say, $x$ axis, which
covers the origin. $\tilde L(0)$ is hence the zero interval
(the one covering the origin) of the point process of macro cell boundary
crossings with the $x$ axis. This process has intensity 
$\lambda_c=\frac{4 \sqrt{\lambda_{BS}p}}{\pi}$
cf~\cite[Equations~$5.7.4$ with $m=2$]{okabe99spatial-tesselations}.
Using the same inverse formula (this time in one dimension)
we obtain:
\begin{align*}
\E_{MU}^0[|L_{macro}(0)|]=\E[|\tilde L(0)|]=\lambda_{c}\E_{cross}^0[|\tilde L(0)|^2]\,
\end{align*}
where $\E_{cross}^0$ corresponds to the Palm probability of the
point process of the macro cell boundary crossings by the horizontal
$x$ axis. 
By \cite[Table~$5.7.2$]{okabe99spatial-tesselations},
$\E_{cross}^0[|\tilde L_{macro}(0)|]^2=\frac{0.804}{\lambda_{BS}p}$. 
Hence, 
$$\E_{MU}^0[|L_{macro}(0|]=\frac{0.804 \times 4}{\pi \times \sqrt{\lambda_{BS}p}}=\frac{3.216 }{\pi \sqrt{\lambda_{BS}p}}.$$ 
Plugging in the expression~\eqref{eqn:intermediate-expression-mean-number-of-mobile-users-served-simultaneously}
we prove the theorem.

 \subsection{Proof of Theorem~\ref{theorem:mean-number-of-SU-served-by-a-macro-BS}}
 \label{subsection:proof-of-mean-number-of-SU-served-by-a-macro-BS}
Let $\E _{macro}^0$ 
denote expectation w.r.t. the Palm probability distribution $\Pro_{macro}^0$ (probability 
given that a macro BS is located at the origin). 
Similarly as in the proof of 
Theorem~\ref{theorem:mean-number-of-MU-served-by-a-macro-BS}:

\footnotesize
\begin{eqnarray}
  \overline{N}_{SU, macro} &=& \lambda_{SU}\E[|V_{het}(X^*)|] \nonumber\\
 &=& \lambda_{SU} \frac{ \E _{macro}^0 [|V_{het}(0)|\times |V_{macro}(0)|]   }{ \E _{macro}^0 [|V_{macro}(0)|]} \nonumber\\
 &=& \lambda_{SU} \lambda_{BS} p \E _{macro}^0[|V_{het}(0)|\times| V_{macro}(0) |]  \label{eqn:basic-expression_number-of-statis-users-served-by-the-macro-BS}
\end{eqnarray}
\normalsize
where the  first equality is by the independence between
$\mathcal{U}_{static}$ and $\Phi$ and  the second by the inverse formula
of Palm calculus. Now, 

\footnotesize
\begin{eqnarray}
 && \E _{macro}^0 [|V_{het}(0)|\times| V_{macro}(0)|] \nonumber\\
 &=& \E _{macro}^0 \bigg( \int_{(x_1,x_2) \in \mathbb{R}^2} \ind  \{ (x_1,x_2) \in V_{macro}(0) \} dx_1 dx_2  \nonumber\\
 && \times \int_{(y_1,y_2) \in \mathbb{R}^2} \ind  \{ (y_1,y_2) \in V_{het}(0) \} dy_1 dy_2 \bigg) \nonumber\\
&=& \E _{macro}^0 \bigg( \int_{(x_1,x_2)}  \int_{(y_1,y_2)} \ind  \{ (x_1,x_2) \in V_{macro}(0), \nonumber\\
&&\hspace{5em} (y_1,y_2) \in V_{het}(0) \} dx_1 dx_2 dy_1 dy_2 \bigg) \nonumber\\
&=& \int_{(x_1,x_2)}  \int_{(y_1,y_2)} \Pro _{macro}^0\{ (x_1,x_2) \in V_{macro}(0), \nonumber\\
&& \hspace{5em} (y_1,y_2) \in V_{het}(0) \} dx_1 dx_2 dy_1 dy_2  \label{eqn:second-expression_number-of-statis-users-served-by-the-macro-BS}
\end{eqnarray}
\normalsize
Given that there is a macro BS at the origin, 
$(x_1,x_2) \in V_{macro}(0)$ and $(y_1,y_2) \in V_{het}(0)$ if and only if these three conditions are satisfied: 
(i) there is no other macro BS in a circle centered at $(x_1,x_2)$ and having radius $\sqrt{x_1^2+x_2^2}$, (ii) there is no 
other macro BS in a circle centered at $(y_1,y_2)$ and having radius $\sqrt{y_1^2+y_2^2}$, and 
(iii) there is no 
other micro BS inside a circle centered at $(y_1,y_2)$ and having radius $r_0$ where $P_{macro}(\sqrt{y_1^2+y_2^2})^{-\beta}=P_{micro}r_0^{-\beta}$, i.e., 
$r_0=(\frac{P_{micro}}{P_{macro}})^{\frac{1}{\beta}}\sqrt{y_1^2+y_2^2}$.

Hence, 
\begin{align*}
&\Pro _{macro}^0\{ (x_1,x_2) \in V_{macro}(0),  (y_1,y_2) \in
  V_{het}(0) \}\\
&=e^{-\lambda_{BS}p \mathcal{A}((x_1,x_2),(y_1,y_2)) -
    \lambda_{BS}(1-p) \pi r_0^2}\,,
\end{align*}
where $r_0=(\frac{P_{micro}}{P_{macro}})^{\frac{1}{\beta}}\sqrt{y_1^2+y_2^2}$. This, combined with 
(\ref{eqn:basic-expression_number-of-statis-users-served-by-the-macro-BS}) and 
(\ref{eqn:second-expression_number-of-statis-users-served-by-the-macro-BS}), proves the theorem.

\subsection{Proof of Theorem~\ref{theorem:approximation-of-static-users-served-by-macro-BS}}
\label{subsection:proof-of-approximation-of-static-users-served-by-macro-BS}

Let us assume that the macro BS closest to the origin is located at a distance $r$ from the origin. Then, the origin will be served by the macro BS 
if and only if there is no micro BS in a circle centered at origin with radius $R$, where $P_{micro}R^{-\beta}=P_{macro}r^{-\beta}$, 
i.e., $R=(\frac{P_{micro}}{P_{macro}})^{1/\beta}r$. The probability that the nearest macro BS to the origin is located at a distance between 
$r$ and $r+dr$ is given by $f(r)dr$ where $f(r)=e^{-\lambda_{BS}p \pi r^2} \lambda_{BS}p \times 2 \pi r$. 
Hence, the probability that a static user located at the origin is served by a macro BS is given by 
$\int_0^{\infty} e^{- \lambda_{BS}(1-p) \pi (P_{micro}/P_{macro})^{2/\beta}r^2} f(r) dr$, which, after simplification, yields 
that the probability that a typical static user 
located at the origin is served by a macro BS is given by 
 $\frac{pP_{macro}^{2/\beta}}{pP_{macro}^{2/\beta}+(1-p)P_{micro}^{2/\beta}}$. This is also the fraction of area over 
 $\mathbb{R}^2$ where SUs are served by macro BSs. 
 
 Now, by the inverse formula of Palm calculus, 
 $\E_{macro}^0 [|V_{het}(0)|]=\frac{pP_{macro}^{2/\beta}}{pP_{macro}^{2/\beta}+(1-p)P_{micro}^{2/\beta}} \times \frac{1}{\lambda_{BS}p}$, 
 from which the proof  follows.

\subsection{Proof of Theorem~\ref{theorem:lower-bound-of-static-users-served-by-macro-BS}}
\label{subsection:lower-bound-of-static-users-served-by-macro-BS}
Let us recall the expression for $\overline{N}_{SU,macro}$ from Theorem~\ref{theorem:mean-number-of-SU-served-by-a-macro-BS}. 
Note that, $\mathcal{A}((x_1,x_2),(y_1,y_2)) \leq \pi (x_1^2+x_2^2)+\pi (y_1^2+y_2^2)$. Hence, 

\begin{eqnarray*}
&& \overline{N}_{SU, macro}\\ 
 & \geq & \lambda_{SU} \lambda_{BS} p \int_{(x_1,x_2) \in \mathbb{R}^2}\int_{(y_1,y_2) \in \mathbb{R}^2} \\
 && e^{-\lambda_{BS}p \pi (x_1^2+x_2^2)-\lambda_{BS}p \pi (y_1^2+y_2^2) - \lambda_{BS}(1-p) \pi r_0^2} \\
 && dx_1 dx_2 dy_1 dy_2 
\end{eqnarray*}
Now, $\int_{(x_1,x_2) \in \mathbb{R}^2} e^{-\lambda_{BS}p \pi (x_1^2+x_2^2)} dx_1 dx_2=\E_{macro}^0 (|V_{macro}(0)|)=\frac{1}{\lambda_{BS} p}$ 
and $\int_{(y_1,y_2) \in \mathbb{R}^2} e^{-\lambda_{BS}p \pi (y_1^2+y_2^2) - \lambda_{BS}(1-p) \pi r_0^2} dy_1 dy_2 =\E_{macro}^0 (|V_{het}(0)|)$. 
Hence, $\overline{N}_{SU, macro} \geq \lambda_{SU} \E_{macro}^0 (|V_{het}(0)|):=\hat{N}_{SU, macro}$.

\subsection{Proof of Theorem~\ref{theorem:upper-bound-of-static-users-served-by-macro-BS}}
\label{subsection:proof-of-upper-bound-of-static-users-served-by-macro-BS}
Note that, $\overline{N}_{SU,macro}=\lambda_{SU} \E_{MU}^0 [|V_{het}(X^*)|]$. But $V_{het}(X^*)$ is a subset 
of macro cell $V_{macro}(0)$ containing the origin, and $\E_{MU}^0 [|V_{macro}(0)|]=\E [|V_{macro}(0)|]=\frac{1.2802}{\lambda_{BS}p}$ 
(as shown in the proof of Theorem~\ref{theorem:mean-number-of-MU-served-by-a-macro-BS}). 
Hence, $\overline{N}_{SU,macro} \leq \frac{1.2802 \lambda_{SU}}{\lambda_{BS}p}$.

\subsection{Proof of Theorem~\ref{theorem:mean-number-of-SU-served-by-a-hetnet-BS}}
 \label{subsection:proof-of-mean-number-of-SU-served-by-a-hetnet-BS}
We consider now the typical SU located at the origin under $\E_{SU}^0$. Using the similar arguments as 
in the proof of Theorem~\ref{theorem:mean-number-of-SU-served-by-a-macro-BS}, we obtain:

\begin{eqnarray}\label{eqn:basic-expression_number-of-statis-users-served-by-the-hetnet-BS}
 \overline{N}_{SU, het}&=&1+\lambda_{SU} \E[|V_{het}(0)|]\nonumber\\
&=&1+\lambda_{SU} \frac{\E _{het}^0 [|V_{het}(0)|^2] }{ \E _{het}^0[|V_{het}(0)|]   }\,,
\end{eqnarray}

where $\E _{het}^0$ denotes the expectation under 
the Palm probability given that there is one (macro or micro) BS at the origin.
Let $\E _{macro}^0$ (resp., $\E _{micro}^0$) be the expectation under the Palm probability distribution given that there is 
one macro (resp., micro)  BS at the origin. Note that, the  
fraction of the macro base stations is $p$, and the rest of base 
stations are micro base stations. Using this fact and using similar arguments as in the proof of 
Theorem~\ref{theorem:mean-number-of-SU-served-by-a-macro-BS}, we can write:

\footnotesize
 \begin{eqnarray}
 && \E _{het}^0 [|V_{het}(0)|^2] \nonumber\\
 &=& p \int_{(x_1,x_2) \in \mathbb{R}^2} \int_{(y_1,y_2) \in \mathbb{R}^2} \nonumber\\
 && \Pro _{macro}^0\{ (x_1,x_2) \in V_{het}(0),  (y_1,y_2) \in V_{het}(0) \} dx_1 dx_2 dy_1 dy_2 \nonumber\\
 && + (1-p)   \int_{(x_1,x_2) \in \mathbb{R}^2} \int_{(y_1,y_2) \in \mathbb{R}^2} \nonumber\\
 && \Pro _{micro}^0\{ (x_1,x_2) \in V_{het}(0),  (y_1,y_2) \in V_{het}(0) \} dx_1 dx_2 dy_1 dy_2 \nonumber\\
&=& p   \int_{(x_1,x_2) \in \mathbb{R}^2} \int_{(y_1,y_2) \in \mathbb{R}^2}  e^{-\lambda_{BS}p\mathcal{A}((x_1,x_2),(y_1,y_2))} \nonumber\\
&&  e^{-\lambda_{BS}(1-p) \mathcal{B}((x_1,x_2),(y_1,y_2))  }  dx_1 dx_2 dy_1 dy_2 \nonumber\\
 && + (1-p) \int_{(x_1,x_2) \in \mathbb{R}^2} \int_{(y_1,y_2) \in \mathbb{R}^2} e^{-\lambda_{BS}(1-p)\mathcal{A}((x_1,x_2),(y_1,y_2))} \nonumber\\
 && e^{-\lambda_{BS}p\mathcal{D}((x_1,x_2),(y_1,y_2))}   dx_1 dx_2 dy_1 dy_2 \label{eqn:second-expression_number-of-statis-users-served-by-the-hetnet-BS}
 \end{eqnarray}
 \normalsize
 
 Now, note that $\E _{het}^0[|V_{het}(0)|]=1/\lambda_{BS}$. This, combined with 
 (\ref{eqn:basic-expression_number-of-statis-users-served-by-the-hetnet-BS}) and 
 (\ref{eqn:second-expression_number-of-statis-users-served-by-the-hetnet-BS}) proves the theorem.

\subsection{Proof of Theorem~\ref{theorem:approximation-of-mobile-users-served-by-hetnet-BS}} 
\label{subsection:proof-of-approximation-of-mobile-users-served-by-hetnet-BS}
As in the proof of Theorem~\ref{theorem:approximation-of-static-users-served-by-macro-BS}, 
$\Pro\{V_{het}(X^*)=V_{het}(0)\} = \frac{pP_{macro}^{2/\beta}}{pP_{macro}^{2/\beta}+(1-p)P_{micro}^{2/\beta}}$. 
The mean volume of the 
zero macro cell (the macro cell containing the origin) is $\E ( | V_{macro}(0)| )=\frac{1.2802}{\lambda_{BS}p}$ as 
in the proof of 
Theorem~\ref{theorem:mean-number-of-MU-served-by-a-macro-BS}, and 
$\E ( |\mathcal{U}_{mobile} ( V_{macro}(0) ) | )=\lambda_{MU} \lambda_L \E ( | V_{macro}(0)| )$. Hence, the proof follows.

 \subsection{Proof of Theorem~\ref{theorem:approximation-of-static-users-served-by-hetnet-BS}} 
 \label{subsection:proof-of-approximation-of-static-users-served-by-hetnet-BS}
 As in the proof of Theorem~\ref{theorem:approximation-of-static-users-served-by-macro-BS}, 
$\Pro\{V_{het}(X^*)=V_{het}(0)\} = \frac{pP_{macro}^{2/\beta}}{pP_{macro}^{2/\beta}+(1-p)P_{micro}^{2/\beta}}$. Now, 
$\Pro\{V_{het}(X^*) \neq V_{het}(0)\}=1-\Pro\{V_{het}(X^*) = V_{het}(0)\}$. 

Similar to the proof of Theorem~\ref{theorem:approximation-of-static-users-served-by-macro-BS}, 
the mean area of a hetnet cell served by a typical macro BS is given by  
$\E_{macro}^0[|V_{het}(0)|]=\frac{pP_{macro}^{2/\beta}}{pP_{macro}^{2/\beta}+(1-p)P_{micro}^{2/\beta}} \times \frac{1}{\lambda_{BS}p}$, 
and similarly $\E_{micro}^0[|V_{het}(0)|]=\frac{(1-p)P_{micro}^{2/\beta}}{pP_{macro}^{2/\beta}+(1-p)P_{micro}^{2/\beta}} \times \frac{1}{\lambda_{BS}(1-p)}$.  

Combining the above results, we prove the theorem.

\subsection{Proof of Lemma~\ref{lemma:linear-combination-convex-increasing-in-xi}}
\label{subsection:proof-of-linear-combination-convex-increasing-in-xi}
Note that, for given values of $p$, $P_{micro}$ and $P_{macro}$, the function $(r_{MU}+\xi r_{SU})$ is an affine, increasing function of $\xi$. 
This proves the lemma since pointwise supremum of affine, increasing functions is convex, increasing.

\subsection{Proof of Lemma~\ref{lemma:linear-combination-optimization-rate-of-mobile-user-decreases-in-xi}}
\label{subsection:proof-of-linear-combination-optimization-rate-of-mobile-user-decreases-in-xi}
Consider any $\kappa>0$. By optimality of $p^*(\xi)$, $P_{micro}^*(\xi)$ and $P_{macro}^*(\xi)$, we obtain:

\footnotesize
\begin{eqnarray*}
&& r_{MU}(p^*(\xi),P_{micro}^*(\xi),P_{macro}^*(\xi)) \\
&& +\xi r_{SU}(p^*(\xi),P_{micro}^*(\xi),P_{macro}^*(\xi)) \\
&\geq& r_{MU}(p^*(\xi+\kappa),P_{micro}^*(\xi+\kappa),P_{macro}^*(\xi+\kappa)) \\
&& +\xi r_{SU}(p^*(\xi+\kappa),P_{micro}^*(\xi+\kappa),P_{macro}^*(\xi+\kappa))
\end{eqnarray*}
\normalsize
and
 \footnotesize
\begin{eqnarray*}
&&r_{MU}(p^*(\xi+\kappa),P_{micro}^*(\xi+\kappa),P_{macro}^*(\xi+\kappa)) \\
&& +(\xi+\kappa) r_{SU}(p^*(\xi+\kappa),P_{micro}^*(\xi+\kappa),P_{macro}^*(\xi+\kappa))\\
&\geq& r_{MU}(p^*(\xi),P_{micro}^*(\xi),P_{macro}^*(\xi)) \\
&& +(\xi+\kappa) r_{SU}(p^*(\xi),P_{micro}^*(\xi),P_{macro}^*(\xi))
\end{eqnarray*}
\normalsize

Adding the above inequalities and cancelling common terms, we obtain 
$r_{SU}(p^*(\xi+\kappa),P_{micro}^*(\xi+\kappa),P_{macro}^*(\xi+\kappa)) \geq r_{SU}(p^*(\xi),P_{micro}^*(\xi),P_{macro}^*(\xi))$, i.e., 
$r_{SU}^*(\xi)$ is increasing in $\xi$. We can prove the other part in a similar way.

\end{document}